\documentclass[11pt]{article}

\usepackage{amsmath, amsthm, amssymb}
\usepackage[margin=1in]{geometry}
\usepackage{framed}
\usepackage{pgfplots}
\pgfplotsset{compat=1.18}

\usepackage{booktabs} 
\usepackage[ruled]{algorithm2e} 
\usepackage{dsfont}
\usepackage{cleveref}
\usepackage{mathtools}
\usepackage{amsthm}
\usepackage{amsmath}
\usepackage{acronym}
\usepackage{nicefrac}
\usepackage{textcomp}
\usepackage{natbib}
\usepackage{makecell}
\usepackage{url}

\usepackage{tikz}

\renewcommand{\Pr}{\mathbb{P}}

\usepackage[]{color-edits}

\addauthor{MM}{teal}

\addauthor{MF}{magenta}

\addauthor{PD}{blue}

\newcommand{\alloc}[0]{a}
\newcommand{\ALG}[0]{A}

\newcommand{\be}{\mathbb{E}}
\newcommand{\bn}{\mathbb{N}}
\newcommand{\br}{\mathbb{R}}

\DeclarePairedDelimiter{\indbrackets}{[}{]}
\newcommand{\ind}[1]{\mathds{1}_{\indbrackets{#1}}}

\newcommand{\ra}{R_{\mathrm{EA}}^\alpha}
\newcommand{\rp}{R_{\mathrm{EP}}^\alpha}
\newcommand{\uton}{U^\star}
\newcommand{\raon}{R_{\mathrm{EA}}^{\alpha,\star}}
\newcommand{\rpon}{R_{\mathrm{EP}}^{\alpha,\star}}
\newcommand{\optut}{\mathrm{OPT}_U}
\newcommand{\optra}{\mathrm{OPT}^\alpha_{\mathrm{EA}}}
\newcommand{\optrp}{\mathrm{OPT}^\alpha_{\mathrm{EP}}}

\newcommand{\comput}{\mathrm{CR}_{U}}
\newcommand{\compra}{\mathrm{CR}_{\mathrm{EA}}}
\newcommand{\comprp}{\mathrm{CR}_{\mathrm{EP}}}

\newcommand{\KL}{\mathrm{KL}}
\newcommand{\Psiq}{\Psi_q}

\DeclareMathOperator*{\argmax}{arg\,max}

\newtheorem{theorem}{Theorem}[section]
\newtheorem{proposition}[theorem]{Proposition}
\newtheorem{lemma}[theorem]{Lemma}

\theoremstyle{definition}
\newtheorem{definition}[theorem]{Definition}
\newtheorem{remark}[theorem]{Remark}

\definecolor{deepteal}{RGB}{31,119,180} 
\definecolor{coral}{RGB}{214,39,40}     
\definecolor{gridgray}{RGB}{180,180,180}

\title{Fair Prophets}
\author{
Paul D\"{u}tting\thanks{Google Research, Z\"{u}rich, Switzerland. Email: \texttt{duetting@google.com}.}
\and
Michal Feldman\thanks{Tel Aviv University, Tel Aviv, Israel. Email: \texttt{mfeldman@tauex.tau.ac.il}.}
\and
Mathieu Molina\thanks{Tel Aviv University, Tel Aviv, Israel. Email: \texttt{mathieu.molina.research@gmail.com}.}
}
\date{}

\begin{document}

\maketitle

\begin{abstract}

We initiate the study of $\alpha$-fair prophet inequalities. This interpolates between utilitarian welfare $(\alpha=0)$, Nash welfare $(\alpha=1)$, and Rawlsian max-min fairness $(\alpha\to\infty)$.
Given the non-linearity of the objective, it matters when the expectation is applied. For instance, for the Rawlsian objective, it matters whether we aim to maximize $\min \be[u_i]$ or $\be[\min u_i]$. We refer to the former as the ex-ante model, and the latter as the ex-post model.

For ex-ante fairness, full distributional knowledge yields a tight competitive ratio of exactly $1/2$ for every $\alpha\ge 0$. Under sample access, $O(n\log n)$ samples per distribution suffice for a constant competitive ratio when $\alpha\in(0,1]$. In contrast, for every $\alpha>1$, no finite number of samples improves upon the trivial $1/n$ guarantee. Thus, unlike in the utilitarian setting, full-information and sample-access prophet inequalities become fundamentally separated.

For ex-post fairness, 
under full information, we obtain a uniform constant ratio for all $\alpha\in(0,1)$, while for every $\alpha>1$ the competitive ratio collapses to $1/n$. In the sample-access model, one sample per distribution suffices for each fixed $\alpha<1$, but no sample budget depending only on $n$ yields a uniform constant guarantee as $\alpha\to 1$. 
Beyond these phase transitions for $\alpha$-fairness, our results open the door to a broader theory of prophet inequalities for non-linear welfare objectives.

\end{abstract}

\section{Introduction}

A recurring challenge across market design, finance, and public policy is the sequential allocation of scarce resources under uncertainty. In these environments, decisions are typically immediate and irrevocable: committing resources to a current request inherently reduces the capacity to serve future opportunities. 
To analyze such problems, the prophet inequality framework has emerged as a cornerstone paradigm \citep{KrengelSucheston1977Semiamarts, SamuelCahn1984}. Unlike worst-case competitive analysis, which can be overly pessimistic, prophet inequalities assume that inputs are drawn from known (or partially known) distributions. 
Formally, there are $n$ rewards $X_1, \ldots, X_n$, drawn independently from probability distributions $F_1, \ldots, F_n$, which are revealed sequentially in an online fashion.
The goal is to choose a fractional allocation $\alloc_1, \ldots, \alloc_n \in [0,1]$ with $\sum_{i \in [n]} \alloc_i \leq 1$ so as to maximize a given performance objective.\footnote{Traditionally, the prophet inequality problem is framed in terms of stopping time, which here corresponds to integral online allocation. Nevertheless, for the utilitarian version, due to the optimality of backward dynamic programming the performance of the online decision maker remains the same when restricting to integral policies (stopping times). The same restriction to integral policies also does not impact the prophet.}

It has been shown that the prophet-inequality paradigm allows for robust approximation guarantees against an all-knowing ``prophet'' who knows the future. 
In particular, the classic result establishes that an online gambler can capture at least $\nicefrac{1}{2}$ of the prophet's expected maximum value when selecting a single item from a sequence, and this bound is known to be tight.
In the decades since, this utilitarian perspective has been vastly extended. Significant results have moved from single-item selection to complex feasibility constraints, including matroids \citep{KleinbergW12}, combinatorial auctions with XOS valuations \citep{DuettingFKL17,Feldman2015Comb}, 
and limited information settings where only a single sample per distribution is available \citep{rubinstein2020optimal,dutting2024online}.

While the utilitarian objective of maximizing aggregate value is a natural objective in settings such as revenue optimization or total welfare, it often fails to reflect the normative and operational constraints of real-world systems. In many environments, decision makers care not only about improving average performance, but also about ensuring that no participant, group, or time period receives an unacceptably poor, or unfair, outcome. 
Such considerations arise in a wide range of applications. Governments allocating emergency resources, such as disaster relief, humanitarian aid, or critical public services, must ensure that no population is left below a basic viability threshold, even when concentrating resources elsewhere would increase total benefit. Similarly, operators of large-scale digital or physical infrastructure often face asymmetric risks: a single severely under-served user, region, or request can cause disproportionate harm, whether in terms of safety, reliability, or public trust. 

Motivated by these considerations, in this paper we study prophet inequalities for the standard one-parameter family of $\alpha$-fair welfare functions. Given a marginal utility vector $u=(u_1,\ldots,u_n)\in\br_+^n$, define for $\alpha \geq 0$ the $\alpha$-fair welfare as
\begin{equation*}
W_\alpha(u)
\coloneqq
\begin{cases}
\left(\frac1n\sum_{i=1}^n u_i^{1-\alpha}\right)^{\frac{1}{1-\alpha}},
& \alpha\neq 1,\\
\left(\prod_{i=1}^n u_i\right)^{1/n},
& \alpha=1 \\
\min_{ i \in [n]} u_i,& \alpha \to \infty
\end{cases}
\end{equation*}
For $\alpha>1$, we use the convention that if some coordinate $u_i=0$, then $W_\alpha(u)=0$.
The parameter $\alpha$ interpolates between standard welfare objectives: $\alpha=0$ recovers utilitarian average welfare, $\alpha=1$ gives Nash social welfare, and the limit $\alpha\to\infty$ gives the Rawlsian objective $\min_i u_i$. Thus, $\alpha$-fairness provides a continuous way to study how prophet inequalities change as the objective shifts from efficiency toward fairness.

\paragraph{Ex-ante and ex-post fairness.}
Since, unlike the standard utilitarian objective $(\alpha = 0)$, the objective we aim to maximize is non-linear, it matters at which point the expectation is applied. For example, for the Rawlsian objective $(\alpha \rightarrow \infty)$ it makes a difference whether we aim to maximize $\min \be[u_i]$ or $\be[\min u_i]$. 

In our first model, the \emph{ex-ante model}, we require that the outcome is equitable in expectation.
That is, the goal is to choose a fractional allocation $\alloc_1,\dots,\alloc_n$ in an online manner so as to maximize $W_\alpha(\be[a \cdot X])$ where $\be[a \cdot X]=\left(\be[\alloc_1X_1],\ldots,\be[\alloc_nX_n]\right)$ is the expected-utility vector induced by the allocation. This objective asks whether each agent, group, or time period receives sufficient utility on average across repeated executions of the allocation process.

In our second model, the \emph{ex-post model}, we require that the outcome is equitable for each realization. Formally, 
the goal is to choose
a fractional allocation $\alloc_1, \ldots, \alloc_n$ 
in an online manner so as to maximize $\be[W_\alpha(a \cdot X)]$, where $a \cdot X=\left(\alloc_1X_1,\ldots,\alloc_nX_n\right)$ is the realized-utility vector induced by the allocation. This objective is more demanding: it requires the allocation to be balanced within each realized instance, rather than only after averaging across different possible realizations.

The order of expectation and welfare is not a cosmetic distinction. In the ex-ante model, an online policy only needs to balance the expected marginal utilities of the agents; randomness can average out across sample paths. In the ex-post model, by contrast, the policy must produce a balanced allocation on each realized instance, while still making irrevocable online decisions. This makes the ex-post objective sensitive not only to which agents receive large values, but also to how much budget remains when later values are revealed. As a consequence, the two models lead to sharply different outcomes and algorithmic behavior. We defer a discussion of further related work to \Cref{app:related-work}.

\subsection{Our Contribution}

We study $\alpha$-fair prophet inequalities in both the ex-ante and ex-post models, under full distributional knowledge and under sample access. Our results highlight a phase transition around $\alpha=1$: for ex-ante fairness, $\alpha\le1$ admits robust sample-based guarantees while $\alpha>1$ does not; for ex-post fairness, $\alpha<1$ admits constant full-information guarantees, whereas $\alpha>1$ collapses to a horizon-dependent guarantee.

\paragraph{Ex-ante $\alpha$-fairness.}
Our first result shows that, with full distributional knowledge, the classic $\nicefrac{1}{2}$ prophet inequality extends to the full $\alpha$-fair family.

\medskip
\noindent
\textbf{Theorem (Tight full-information ex-ante guarantee; \Cref{thm:alpha-full-info-half}).}
For every $\alpha\ge0$, the ex-ante $\alpha$-fair competitive ratio is $1/2$.

\medskip

The proof is not specific to the closed form of $W_\alpha$. We isolate a class of welfare functions that are symmetric, concave, and positively homogeneous, and show that every such welfare function admits a $\nicefrac{1}{2}$ ex-ante prophet inequality.

The lower bound is based on a geometric strengthening of the classic utilitarian prophet inequality. We show that the online decision maker can simultaneously approximate the prophet's entire attainable utility region by a factor $\nicefrac12$, not merely every linear objective separately. This follows from the characterization of closed convex sets by their support functions, together with a Hahn--Banach separation argument. Monotonicity and homogeneity of the welfare function then transfer this geometric containment directly to any such non-linear welfare objective, including the full $\alpha$-fair family. The matching upper bound is already achieved by a two-agent long-shot instance. The full-information result extends the classical utilitarian guarantee to a broad class of non-linear objectives.

The sample-access setting is more delicate. For $\alpha\in(0,1]$, we show that a polynomial number of samples per distribution suffices for a constant-factor approximation.

\medskip
\noindent
\textbf{Theorem (Sample-based ex-ante guarantee; \Cref{thm:sample-alpha-small}).}
For every $\alpha\in(0,1]$, there is an algorithm using $O(n\log n)$ samples per distribution that achieves a $\frac{1-e^{-1}}{128e}$ constant  competitive ratio for the ex-ante $\alpha$-fair objective.

\medskip

The algorithm first learns a relaxed prophet allocation in terms of per-agent target activation probabilities, and then implements those probabilities online via rank thresholds computed from fresh samples. The proof combines a prophet relaxation based on marginal utility curves, relative tail approximation from samples, and a rank-threshold implementation argument.

This positive result cannot be extended beyond $\alpha=1$. When $\alpha>1$, the welfare objective becomes sufficiently bottleneck-sensitive that any finite amount of samples cannot identify which agent must receive most of the budget.

\medskip
\noindent
\textbf{Theorem (Sample impossibility for ex-ante $\alpha>1$; \Cref{thm:sample-impossibility-alpha-large}).}
For all $\alpha >1$ and the ex-ante objective with unknown distributions, the competitive ratio is $\nicefrac{1}{n}$ for any finite number of samples $m \in \mathbb{N}$.

\medskip

In particular, this yields the first family of settings ($\alpha >1$), where a constant-competitive approximation is achievable with full information, but this guarantee does not translate to a similar guarantee with sample access.

To prove the theorem, we construct $n$ environments, each with a different ``bad'' agent. In each environment, any near-optimal algorithm must concentrate almost all allocation probability on the bad agent, but with only finitely many samples the environments are information-theoretically indistinguishable, as their pairwise total variation distance is small. Hence, any learning algorithm must behave nearly identically across environments and cannot systematically identify the bad agent, forcing a spread allocation and yielding an upper bound of $1/n$.

\paragraph{Ex-post $\alpha$-fairness.}
The ex-post model exhibits a different phase transition. The case $\alpha=1$ is trivial, and a competitive ratio of $1$ is achieved via a uniform allocation $\alloc_i=1/n$ independent of the instance. For $\alpha\in(0,1)$, we prove that a universal constant competitive ratio, independent of $\alpha$, is possible.

\medskip
\noindent
\textbf{Theorem (Universal full-information ex-post guarantee; \Cref{thm:cr_uniform}).}
For every $\alpha\in (0,1)$, the ex-post $\alpha$-fair competitive ratio is at least $e^{-\pi^2/6}>0.193$. 

\medskip

First, under an appropriate change of measure, it is equivalent to lower bound the ratio of expectations under the original law and to lower bound the expectation of the ratio under a biased law. In order to lower bound the biased expectation of the ratio, we need for the online algorithm to perform well on every realization. We formalize this intuition, by lower bounding the realized ratio by some Kullback-Leibler divergence that measures how far apart the online and prophet allocation are. This suggests to design a policy that never runs out of budget, and we do so by allocating a fraction of the remaining budget at each time step through a function that estimates the posterior allocation given the current realizations.

Similarly to the ex-ante setting, we are able to obtain performance guarantees through sample access. For every fixed $\alpha \in (0,1)$, a single sample per distribution is enough to obtain a constant that depends on $\alpha$.

\medskip
\noindent
\textbf{Theorem (Single-sample ex-post guarantee for fixed $\alpha<1$; \Cref{thm:single_sample_CR}).}
For every $\alpha\in(0,1)$, the single-sample algorithm in \Cref{alg:single-sample-spike-fill} achieves competitive ratio at least $1/(4\cdot 3^{\alpha/(1-\alpha)})$.

\medskip

This algorithm mixes two policies. One is the classical single-sample threshold algorithm from \citet{rubinstein2020optimal}, which captures instances where value is concentrated in one coordinate. The other is a sample-target filling policy, which treats the sampled aggregate as a proxy for the realized aggregate and allocates proportionally until the budget is exhausted, this captures realizations where the reward is spread out across coordinates. A moment decomposition shows that these two regimes together capture a constant fraction of the prophet value for fixed $\alpha$.

However, the dependence on $\alpha$ in the previous theorem deteriorates as $\alpha\to1$. This raises the question on whether a universal guarantee, similar to the full information case, is possible.  

We prove that, in contrast to the ex-ante setting where a constant was also achievable with a number of samples independent of $\alpha$, in the ex-post setting such a universal approximation guarantee cannot be achieved with a fixed finite number of samples. Namely, as formalized by the following theorem, for every sample budget depending only on $n$, there is a sequence of parameters $\alpha_n\to1$ for which the competitive ratio goes to zero.

\medskip
\noindent
\textbf{Theorem (No uniform sample guarantee for ex-post $\alpha<1$; \Cref{thm:sample-impossibility-stated}).}
For every fixed $n\geq2$ and every finite sample budget $m$ that depends only on $n$, the best sample access ex-post competitive ratio converges to $1/n$ as $\alpha \to 1^-$.

\medskip

The proof uses a family of environments with rare positive values that increase geometrically up to an unknown last active agent, after which all the values are constantly $0$. With high probability, all training samples are zero. On the path where all active agents take their positive values, the algorithm must spread its budget across the possible unknown last agents cutoffs, while the prophet allocation concentrates on the last active agent. Choosing $\alpha$ sufficiently close to $1$ makes these rare paths dominate the expected prophet value. An averaging argument then yields the sharp limiting bound $1/n$. 
Thus, for ex-post fairness below $\alpha=1$, the sample complexity and the competitive ratio cannot both be made uniform as $\alpha$ approaches Nash social welfare.

Finally, we show that, when $\alpha>1$, even full distributional knowledge is not enough for a constant competitive ratio.

\medskip
\noindent
\textbf{Theorem (Full-information ex-post hardness for $\alpha>1$; \Cref{thm:alpha-expost-hardness}).}
For every $\alpha>1$ and every horizon $n$, the ex-post competitive ratio is exactly $1/n$. 

\medskip

This result shows that ex-post $\alpha$-fairness undergoes a sharp collapse once $\alpha>1$: the online decision maker cannot obtain a constant fraction of the prophet, even with full knowledge of the distributions.

To establish the matching $1/n$ impossibility, we construct independent values on a geometric scale. On the event that agent $i$ is the last to take its low value, its utility upper bounds the minimum utility. 
Independence of the future and adaptiveness of the online allocation lets us compute the probability that agent $i$ is the last to take its low value. The budget constraint bounds the total online value, while each possible last low value contributes approximately equally to the prophet benchmark. Letting the ratio between consecutive scales grow yields the sharp upper bound. The matching lower bound follows immediately from monotonicity and homogeneity under uniform allocation.

This result improves upon the impossibility result of~\cite{manshadi2023fair} who study the special case $\alpha \to \infty$, who established an $O(\nicefrac{1}{\log n})$ upper bound for correlated values. Our result strengthens this in two respects: it improves the bound to $O(\nicefrac{1}{n})$ and applies even to independent values.

\begin{figure}
    \centering
    \makebox[\linewidth][l]{%
    \hspace*{-0.05\linewidth}%
    \begin{minipage}{\linewidth}
    \begin{minipage}[t]{0.49\linewidth}
    \centering
    \begin{tikzpicture}
\begin{axis}[
    width=0.86\linewidth, height=0.75\linewidth,
    xmin=0, xmax=2.35, ymin=0, ymax=0.72,
    axis lines=left, xlabel={$\alpha$: Ex-ante model}, ylabel={CR},
    xtick={0,1,2.25}, xticklabels={$0$,$1$,$\infty$},
    ytick={0.03,0.32,0.50,0.66},
    yticklabels={$1/n$,$\Omega(1)$,$1/2$,$1$},
    tick label style={font=\small}, label style={font=\small}, unbounded coords=jump, clip=false]
    \addplot[blue, line width=2.6pt] coordinates {(0,0.50) (2.30,0.50)};
    \addplot[blue, line width=2.0pt, densely dashed] coordinates {(0,0.52) (2.30,0.52)};
    \addplot[red, line width=2.4pt] coordinates {(0,0.32) (1.00,0.32) (nan,nan) (1.03,0.03) (2.30,0.03)};
    \addplot[red, line width=2.0pt, densely dashed] coordinates {(0,0.47) (1.00,0.47) (nan,nan) (1.03,0.055) (2.30,0.055)};
\end{axis}
\end{tikzpicture}
    \end{minipage}
    \hfill
    \begin{minipage}[t]{0.49\linewidth}
    \begin{tikzpicture}
\begin{axis}[
    width=1.0\linewidth, height=0.75\linewidth,
    xmin=0, xmax=2.35, ymin=0, ymax=0.72,
    axis lines=left, xlabel={$\alpha$: Ex-post model}, ylabel=\empty,
    xtick={0,1,2.25}, xticklabels={$0$,$1$,$\infty$},
    ytick={0.03,0.50,0.66},
    yticklabels={$1/n$,$1/2$,$1$},
    tick label style={font=\small}, label style={font=\small}, unbounded coords=jump, clip=false,
    legend style={draw=none, fill=white, fill opacity=0.78, text opacity=1,
      font=\small, at={(0.48,0.6)}, anchor=north west, row sep=1pt, column sep=6pt, legend columns=1, cells={anchor=west}},
    legend image post style={scale=1.05}]
    \addplot[blue, line width=2.4pt] coordinates {(0.00,0.32) (1.00,0.32) (nan,nan) (1.03,0.03) (2.30,0.03)};
    \addlegendentry{full information: lower}
    \addplot[blue, line width=2.0pt, densely dashed] coordinates {(0.00,0.50) (1.00,0.50) (nan,nan) (1.03,0.055) (2.30,0.055)};
    \addlegendentry{full information: upper}
    \addplot[red, line width=2.4pt, smooth] coordinates {(0.00,0.22) (0.20,0.19) (0.45,0.14) (0.70,0.08) (0.95,0.03) (nan,nan) (1.03,0.03) (2.30,0.03)};
    \addlegendentry{samples: lower}
    \addplot[red, line width=2.0pt, densely dashed, smooth] coordinates {(0.00,0.35) (0.20,0.31) (0.45,0.25) (0.70,0.19) (0.99,0.04) (nan,nan) (1.03,0.04) (2.30,0.04)};
    \addlegendentry{samples: upper}
    \addplot[only marks, mark=*, mark size=3.9pt, color=blue] coordinates {(1,0.66)};
    \addplot[only marks, mark=o, mark size=6.5pt, mark options={red,line width=1.3pt,fill=none}] coordinates {(1,0.66)};
\end{axis}
\end{tikzpicture}
    \end{minipage}
    \end{minipage}%
    }
    \caption{Qualitative summary of the results for the ex-ante and ex-post models. Blue curves correspond to full-information guarantees, and red curves correspond to sample-access guarantees. The curves are schematic; for every fixed finite sample budget, the ex-post sample guarantee approaches $1/n$ as $\alpha\uparrow1$, whereas the ratio at $\alpha=1$ is $1$.}
    \label{fig:exante-expost-summary}
\end{figure}
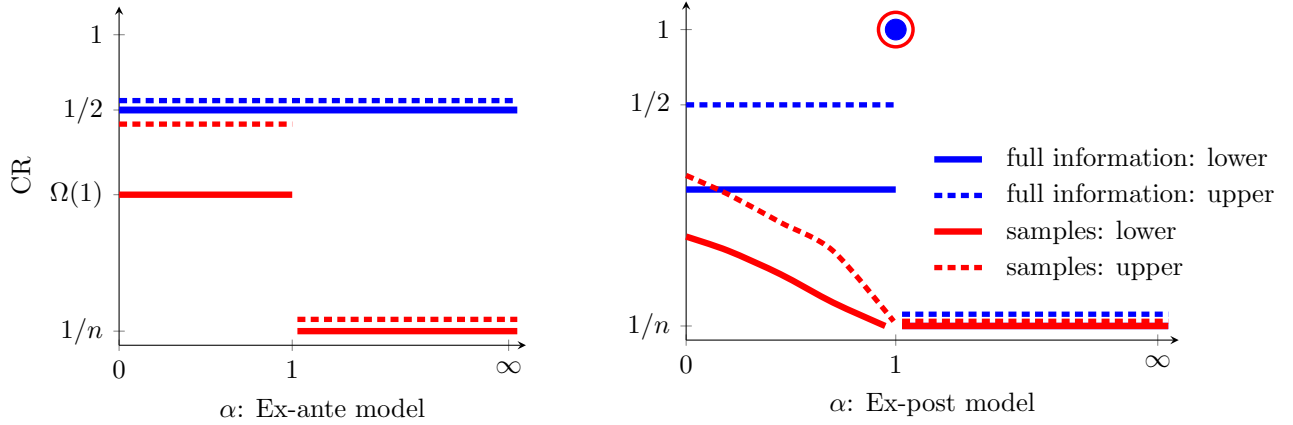

\paragraph{Summary of results.}
Overall, the ex-ante and ex-post models exhibit different phase transitions, as summarized qualitatively in \Cref{fig:exante-expost-summary}. In the ex-ante setting, full-information constant competitive ratios are achievable for every $\alpha\ge0$, but the sample-access model undergoes a sharp transition at $\alpha=1$: finite samples suffice for $\alpha\le1$, whereas for $\alpha>1$ they cannot improve over the trivial $1/n$ guarantee. In the ex-post setting, the transition is already visible under full information, since no constant competitive ratio is possible for $\alpha>1$. For $\alpha\in(0,1)$, constant full-information guarantees are possible, but finite-sample guarantees are necessarily non-universal as $\alpha$ approaches $1$. The endpoint $\alpha=1$ is an exception: the uniform allocation which is instance-independent and sets $\alloc_i = \nicefrac{1}{n}$ is pointwise optimal for the ex-post Nash objective, and therefore achieves competitive ratio $1$ without any distributional knowledge.

The qualitatively different behavior on the two regimes ($\alpha < 1$ and $\alpha > 1$) reflects the influence of $\alpha$, which governs the substitutability of utilities across agents.  
When $\alpha<1$, the welfare function is a positive-order power mean: low utilities are penalized, but they can still be compensated by sufficiently high utilities elsewhere. Thus the objective remains efficiency-like, with a fairness correction. 
When $\alpha>1$, the welfare function becomes a negative-order power mean: the value is governed by the smallest coordinates, and large utilities for well-served agents provide little compensation for a poorly served one. In this regime, the objective is bottleneck-sensitive and approaches the Rawlsian minimum as $\alpha\to\infty$.

\section{Model for Online Resource Allocation} \label{sec:utilitarian}
We consider a resource allocation setting in which a decision maker observes in an online fashion a sequence of integrable, independent, non-negative random variables $X_1,\dots,X_n$. At each time $i\in[n]$, upon observing $X_i$, the decision maker selects an allocation level $\alloc_i\ge 0$, yielding instantaneous utility $\alloc_i X_i$. The total amount allocated over the horizon is constrained by a fixed budget. We normalize this budget to $1$, so that every feasible allocation satisfies $\sum_{i \in [n]} \alloc_i \leq 1$.\footnote{For the purpose of evaluating competitive ratios, scaling the budget uniformly scales both the optimal online and offline rewards by the same factor, which cancels out in the ratio.}

At time $i$ a fractional online algorithm $A$ takes as input $X_1,\dots,X_i$, and outputs an allocation $ A_i(X_1,\dots,X_i) \in [0,1]$. The algorithm must satisfy $\sum_{i \in [n]} A_i(X_1,\dots,X_i)  \leq 1$. 
When clear from the context, we will simply refer to $A_i(X_1,\dots,X_i)$ as $a_i$. 
We denote by $\mathcal{A}_n$ the set of all such online algorithms. A sub-class of $\mathcal{A}_n$ is $\mathcal{A}_n^{\mathrm{int}}$, the class of integral online algorithms, where for all $i \in [n]$, $A_i(X_1,\dots,X_i) \in \{0,1\}$. We allow the online decision maker to randomize ex ante over feasible online algorithms. Since the objectives we consider are concave in the allocation, such randomization does not improve over the corresponding averaged fractional algorithm.

We compare the performance of online algorithms to that of an omniscient prophet who observes the entire sequence $(X_1,\dots,X_n)$ before making any allocation decisions subject to the same feasibility constraints. The class of admissible offline  (look-ahead) algorithms is defined as $\mathcal{L}_n$.

\subsection{The Classic Prophet Inequality for Utilitarian Welfare}

We recall the classical utilitarian prophet inequality, which serves as the $\alpha=0$ benchmark. For an online algorithm $A$, define its expected total utility by $U(A) \coloneqq 
\mathbb{E}\left[\sum_{i=1}^n A_i(X_1,\dots,X_i) X_i\right]
=
\mathbb{E}\left[\sum_{i=1}^n a_i X_i\right]$. 
By linearity, the ex-ante and ex-post formulations coincide for this objective. Let
\begin{equation*}
 \uton \coloneqq \sup_{A\in\mathcal{A}_n} U(A),
 \qquad
 \optut \coloneqq \sup_{A\in\mathcal{L}_n} U(A)
 =
 \mathbb{E}\!\left[\max_{i\in[n]} X_i\right]
\end{equation*}
denote the optimal online and prophet values. The optimal online policy can be computed by dynamic programming~\citep{Chow1971GreatET}; in particular, it is without loss of optimality to restrict to integral stopping policies.
The classical prophet inequality of \citet{KrengelSucheston1977Semiamarts} gives the tight worst-case comparison
\begin{equation*}
 \comput \coloneqq
 \inf_{X_1,\dots,X_n}\frac{\uton}{\optut}
 =
 \frac{1}{2}.
\end{equation*}
This full-information guarantee assumes knowledge of the distributions of the $X_i$. In the sample-access model, \citet{rubinstein2020optimal} showed that the same $1/2$ ratio is achievable with only one sample from each distribution before the online phase.

\section{Ex-Ante $\alpha$-Fairness}\label{sec:exante-alpha}

We start by studying the ex-ante model. For an allocation algorithm $\ALG$, the ex-ante $\alpha$-fair objective is defined as
\begin{equation*}
\ra(\ALG)
\coloneqq W_\alpha(\be[a \cdot X]).
\end{equation*}
We denote the optimal online and prophet values by
\begin{equation*}
    \raon
    \coloneqq
    \sup_{\ALG\in\mathcal A_n}\ra(\ALG),
    \qquad
    \optra
    \coloneqq
    \sup_{\ALG\in\mathcal L_n}\ra(\ALG).
\end{equation*}
The fixed-horizon and all-horizon competitive ratios are defined, respectively, as
\begin{equation*}
    \compra^n(\alpha)
    \coloneqq
    \inf_{X_1,\ldots,X_n}
    \frac{\raon}{\optra}, \qquad  \compra(\alpha)
    \coloneqq
    \inf_{n\geq2}\compra^n(\alpha).
\end{equation*}

The $\alpha$-fair $W_\alpha$ ex-ante objective fundamentally alters optimal behavior compared to the utilitarian benchmark. Consider a deterministic two-item instance: $X_1 = 1$ and $X_2 = 2$. A utilitarian policy (both offline and online) greedily selects only $X_2$, leaving the first agent with a utility of $0$. In contrast, the optimal ex-ante policy for $\alpha\to \infty$ must balance allocations fractionally (or via randomization): it assigns weight $2/3$ to the first item and $1/3$ to the second, yielding an optimal ex-ante value of $2/3$ for both online and offline algorithms. 

\subsection{A Tight $\nicefrac{1}{2}$ Competitive Ratio with Full Information}
\label{subsec:exante-alpha-full-information}

The full-information result is not specific to the closed form of $W_\alpha$. We first isolate the structural properties of the welfare function that are sufficient for the upper and lower bounds to match.

\begin{definition}
We define a function $f:\br_+^n\to\br_+$ to be $\mathrm{SCH}$ if it is symmetric, concave, and positively homogeneous of degree one, that is, $f(\lambda x)=\lambda f(x)$ for every $\lambda\geq0$ and every $x\in\br_+^n$.  
\end{definition}

We next prove a general upper bound for $\mathrm{SCH}$ functions for a horizon $n=2$. See \Cref{app:alpha-general-upper} for the proof. 
\begin{proposition}\label{prop:alpha-general-upper}
Let $f:\br_+^2\to\br_+$ be a nonzero $\mathrm{SCH}$ welfare function. Then the full-information ex-ante competitive ratio for the welfare objective $f$ is at most $\nicefrac{1}{2}$.
\end{proposition}

The lower bound uses a different, and more geometric, set of assumptions. The proof relies only on the classical weighted utilitarian prophet inequality and a convex support-function argument. It generalizes the argument of \cite{HillKennedy1990} for the case $\alpha=\infty$.

\begin{proposition}\label{prop:alpha-general-lower}
Let $n \geq 1$, and $f:\br_+^n\to\br_+$ be coordinatewise monotone and positively homogeneous of degree one. Then, for the ex-ante welfare objective $f(\be[\alloc \cdot X])$, the corresponding competitive ratio is lower bounded by $1/2$. 
\end{proposition}

\begin{proof}
For an algorithm $\ALG\in\mathcal A_n$ or $\ALG \in \mathcal L_n$, recall that  $u(\ALG)
    =
    \left(
    \be[\alloc_1X_1],\ldots,\be[\alloc_nX_n]
    \right)\in\br_+^n$. We define the range of the agent payoffs induced by online and offline algorithms as $\mathcal C
    \coloneqq
    \{u(\ALG):\ALG\in\mathcal A_n\}$ and $\mathcal P
    \coloneqq
    \{u(\ALG):\ALG\in\mathcal L_n\}$ respectively. Both of those sets are downward closed, as the feasibility condition $\sum_i \alloc_i \leq 1$ implies that we can always choose to allocate $0$ instead of $a_i$. 
Both sets are also convex. Indeed, if $u(\ALG),u(\ALG')\in\mathcal C$, then the online algorithm that draws an independent coin before the process starts, runs $\ALG$ with probability $\lambda$, and runs $\ALG'$ with probability $1-\lambda$, is feasible and induces the utility vector $\lambda u(\ALG)+(1-\lambda)u(\ALG')$. The same argument applies to $\mathcal P$. 

For $w\in\br_+^n$, define
\begin{equation*}
    h(w)
    \coloneqq
    \sup_{u\in\mathcal C}\langle w,u\rangle,
    \qquad
    h^\star(w)
    \coloneqq
    \sup_{v\in\mathcal P}\langle w,v\rangle.
\end{equation*}
The quantity $h(w)$ is precisely 
the optimal online utilitarian value for the scaled sequence 
$w_1X_1, \ldots, w_nX_n$. Similarly, $h^\star(w)$ is the corresponding prophet value. In terms of convex analysis, the functions $h$ and $h^\star$ represent the support functions of the convex sets $\mathcal{C}$ and $\mathcal{P}$. The classical prophet  inequality therefore gives $ h(w)
    \geq
    \frac12 h^\star(w)$ for all  $w\in\br_+^n$.

We claim that $\frac12\mathcal P
    \subseteq
    \overline{\mathcal C}$, where $\overline{\mathcal C}$ is the closure of $\mathcal C$.
Indeed, fix $v\in\mathcal P$. By the inequality between $h$ and $h^\star$, for every $w\in\br_+^n$,
\begin{equation*}
    \left\langle w,\frac v2\right\rangle
    \leq
    \frac12 h^\star(w)
    \leq
    h(w).
\end{equation*}
Suppose, toward a contradiction, that $v/2\notin\overline{\mathcal C}$. The set $\overline{\mathcal C}-\br_+^n$ is convex because it is the sum of two convex sets, namely $\overline{\mathcal C}$ and $-\br_+^n$. It is also closed. Indeed, $\overline{\mathcal C}$ is compact since every expected utility vector is contained in $\prod_i[0,\be[X_i]]$. If $z^m=u^m-r^m\to z$ with $u^m\in\overline{\mathcal C}$ and $r^m\in\br_+^n$, then, after passing to a subsequence, $u^m\to u\in\overline{\mathcal C}$. Hence $r^m=u^m-z^m\to u-z$, and since $\br_+^n$ is closed, $u-z\in\br_+^n$. Thus $z=u-(u-z)\in\overline{\mathcal C}-\br_+^n$. Moreover, since $\overline{\mathcal C}$ is downward closed, if $v/2\in\overline{\mathcal C}-\br_+^n$, then $v/2=u-r\le u$ for some $u\in\overline{\mathcal C}$ and $r \in \br_+^n$, which would imply $v/2\in\overline{\mathcal C}$. Thus $v/2\notin\overline{\mathcal C}-\br_+^n$. Thus, by the Hahn-Banach separation theorem, there exists $w\in\br^n$ such that
\begin{equation*}
    \left\langle w,\frac v2\right\rangle
    >
    \sup_{z\in\overline{\mathcal C}-\br_+^n}\langle w,z\rangle.
\end{equation*}
Since $\overline{\mathcal C}-\br_+^n$ is unbounded in every negative coordinate direction, the right-hand side is finite only if $w\in\br_+^n$. For such $w$,
\begin{equation*}
    \sup_{z\in\overline{\mathcal C}-\br_+^n}\langle w,z\rangle
    =
    \sup_{u\in\overline{\mathcal C}}\langle w,u\rangle
    =
    h(w),
\end{equation*}
overall yielding the contradiction $\langle w,v/2\rangle  \leq h(w) < \langle w,v/2\rangle$. This proves the claim.

Now fix $v\in\mathcal P$ and $\gamma\in(0,1)$. Since $v/2\in\overline{\mathcal C}$, there exists a sequence $u^m\in\mathcal C$ with $u^m\to v/2$. For all sufficiently large $m$, we have $u^m\geq\gamma v/2$ coordinatewise. Since $\mathcal C$ is downward closed, this implies $\gamma v/2\in\mathcal C$. Therefore, using homogeneity,
\begin{equation*}
    \sup_{u\in\mathcal C} f(u)
    \geq
    f\left(\gamma\frac v2\right)
    =
    \frac{\gamma}{2}f(v).
\end{equation*}
Letting $\gamma \to 1$ and then taking the supremum over $v\in\mathcal P$ gives
\begin{equation*}
    \sup_{u\in\mathcal C} f(u)
    \geq
    \frac12 \sup_{v\in\mathcal P} f(v). \qedhere 
\end{equation*}
\end{proof}

We can now specialize the two general propositions to the $\alpha$-fair welfare functions.

\begin{theorem}\label{thm:alpha-full-info-half}
For every $\alpha \geq 0$, the ex-ante $\alpha$-fair competitive ratio is $\compra(\alpha)=\frac12$.
\end{theorem}

\begin{proof}
It suffices to verify that $W_\alpha$ is $\mathrm{SCH}$. Symmetry and homogeneity are immediate from the definition of $W_\alpha$. Concavity follows from the concavity of power means of order at most $1$; for $\alpha\neq1$ this is the power mean of order $1-\alpha\leq1$, while for $\alpha=1$ it is the geometric mean. Thus $W_\alpha$ is $\mathrm{SCH}$. By \Cref{lem:sch-properties}, $W_\alpha$ is monotone. Therefore, \Cref{prop:alpha-general-lower} gives the lower bound, and \Cref{prop:alpha-general-upper} gives the matching upper bound for $n=2$. 
\end{proof}

\subsection{An $O(n\log n)$-Sample Constant Competitive Ratio Algorithm with $\alpha\in(0,1]$}
\label{subsec:sample-alpha-small}

We next consider the case where the decision maker does not have any a priori knowledge about the instance except the number of items $n$, and only has sample access. This is the same perspective studied in \citet{rubinstein2020optimal} where they show that a single sample per distribution achieves a $1/2$ competitive ratio for the classic case $\alpha=0$.

We use the following sample model. The distributions are unknown, but the decision maker has access to $m$ independent samples from each distribution before the online phase begins. We write $\widetilde X_{i,k}$ for the $k$-th sample from distribution $F_i$, and the learning algorithm may use an independent random seed in addition to these samples.

The full-information proof above is existential and geometric: it shows that
the online attainable utility region contains a scaled copy of the prophet
region, but it does not by itself produce a distribution-free learning
procedure. In the sample-access model we therefore use a different approach.
We first learn, from samples, a relaxed offline target vector of marginal
activation probabilities. We then implement these target probabilities online
using fresh samples and rank thresholds. The main difficulty is to ensure that
errors in the learned marginal utilities do not destroy the non-linear
$\alpha$-fair objective.

We prove the following constant competitive ratio result for a polynomial number of samples.

\begin{algorithm}[t]
\SetAlgoNoLine
\KwIn{Sample access to $F_1,\ldots,F_n$, online observations $X_1,\ldots,X_n$, and $\alpha\in(0,1]$.}
\KwOut{A feasible online allocation $a_1,\ldots,a_n$.}

Set
\begin{equation*}
    \delta=\frac1{20n},
    \qquad
    m_{\rm tr}
    =
    \left\lceil
    C\delta^{-1}\bigl(\log(1/\delta)+\log n\bigr)
    \right\rceil,
    \qquad
    m_{\rm imp}
    =
    \left\lceil\frac2\delta\right\rceil-1.
\end{equation*}

Draw $m_{\rm tr}$ training samples from each distribution and using $\widehat G_i$ from \Cref{eq:empirical_curves}\;

\If{$\alpha\in(0,1)$}{
    Compute $\widehat c$ from \Cref{eq:learn-chat-alpha-small}\;
}
\If{$\alpha=1$}{
    Compute $\widehat c$ from \Cref{eq:learn-chat-nash}\;
}

Set $d_i=\delta+\widehat c_i$ for all $i \in [n]$.

Draw fresh implementation samples $\widetilde X_{i,1},\ldots,\widetilde X_{i,m_{\rm imp}}$ from each $F_i$\;

Set $k_i=\left\lfloor (m_{\rm imp}+1)d_i\right\rfloor$ for all $i \in [n]$

\For{$i=1,\ldots,n$}{
Observe $X_i$\;

Declare agent $i$ eligible if $X_i$ ranks among the top $k_i$ values of $\{X_i,\widetilde X_{i,1},\ldots,\widetilde X_{i,m_{\rm imp}}\}$, using symmetric random tie-breaking

\If{agent $i$ is eligible}{
Allocate the entire remaining budget to agent $i$ and stop\;
}
}

If no agent is eligible, allocate nothing\;

\caption{Learned Rank-Threshold Algorithm}
\label{alg:quantile-rank-threshold}
\end{algorithm}

\begin{theorem}\label{thm:sample-alpha-small}
For every $\alpha\in(0,1]$, there exists an algorithm $A$ using $O(n\log n)$ samples per distribution that achieves
\begin{equation*}
    W_\alpha(u(\ALG))
    \geq
    {\frac{e-1}{128e^2}}\optra.
\end{equation*}
\end{theorem}

The remainder of the subsection proves the theorem for the learning algorithm,  \Cref{alg:quantile-rank-threshold}.

\subsubsection{Prophet Relaxation Upper Bound}
\label{subsubsec:relaxation}

As a first step, we provide an ex-ante relaxation upper bound of the prophet benchmark. This relaxation is defined in terms of the best utility that a single agent can obtain when assigned an expected budget. 

 For each distribution $F_i$, we denote the survival function as $\overline F_i(t)=\Pr(X_i>t)$.
For each agent $i \in [n]$ and an allocation budget $b \geq 0$, we define 
\begin{equation}
\mu_i(b)
    \coloneqq
    \int_0^\infty \min\{\overline F_i(t),b\}\,dt.    
\end{equation}
This quantity $\mu_i(b)$ is the best expected utility one can extract from
agent $i$ using expected allocation probability at most $b$; it is obtained by spending this probability mass on the upper tail of $X_i$.

The next lemma provides an upper bound on the best marginal utility achievable with an expected budget of $b$. 
\begin{lemma}\label{lem:upper_bound_marginal}
For $i \in [n]$, the functions $\mu_i$ are nondecreasing and concave, and for any measurable allocation rule $\varphi:\br_+\to[0,1]$ satisfying $\be[\varphi(X_i)]\leq b$, we have $\be[\varphi(X_i)X_i]\leq\mu_i(b)$.
\end{lemma}

\begin{proof}
For every fixed $t$, the map $b\mapsto \min\{\overline F_i(t),b\}$ is nondecreasing and concave. Integrating over $t$ preserves both properties. By Fubini's theorem, and using that $\varphi \leq 1$,
\begin{equation*}
    \be[\varphi(X_i)X_i]
    =
    \int_0^\infty \be[\varphi(X_i)\ind{X_i>t}]\,dt
    \leq
    \int_0^\infty \min\{\be[\varphi(X_i)],\overline F_i(t)\}\,dt
    \leq
    \mu_i(b). \qedhere
\end{equation*} 
\end{proof}

Now, define the corresponding relaxed prophet benchmark
\begin{equation*}
    \Gamma_\alpha(F)
    \coloneqq
    \sup_{b\in\br_+^n:\,\sum_i b_i\leq1}
    W_\alpha(\mu_1(b_1),\ldots,\mu_n(b_n)).
\end{equation*}

\begin{proposition}\label{prop:relaxed-prophet-upper}
For every $\alpha\in(0,1]$, $\optra\leq\Gamma_\alpha(F)$.
\end{proposition}

\begin{proof}
Fix any prophet allocation rule $\ALG\in\mathcal L_n$, and set $b_i=\be[\alloc_i]$. The budget constraint gives $\sum_i b_i\leq1$. Let $\varphi_i(x)=\be[\alloc_i\mid X_i=x]$. Then $\be[\varphi_i(X_i)]=b_i$ and by \Cref{lem:upper_bound_marginal}
\begin{equation*}
    \be[\alloc_iX_i]=\be[\varphi_i(X_i)X_i]\leq\mu_i(b_i).
\end{equation*}
 Since $W_\alpha$ is monotone, we obtain
\begin{equation*}
    W_\alpha(u(\ALG))
    \leq
    W_\alpha(\mu_1(b_1),\ldots,\mu_n(b_n))
    \leq
    \Gamma_\alpha(F).
\end{equation*}
Taking the supremum over prophet algorithms proves the claim.
\end{proof}

\subsubsection{Learning an Offline Allocation}
\label{subsubsec:learning-relaxation}

We now show that with a sufficient number of samples we can learn a budget profile whose relaxed utility captures a constant fraction of $\Gamma_\alpha(F)$.

First, we can learn a truncated approximation of $\overline{F_i}$. Let $\overline{F}^{(m)}$ be the empirical survival function obtained from $m$ i.i.d. samples.

\begin{lemma}\label{lem:relative-tail-approx}
There is a universal constant $C>0$, such that for $\delta \in (0,1)$ and $ m\geq C\delta^{-1}(\log(1/\delta)+\log n)$ samples per distribution, the following event holds with probability at least $3/4$:
\begin{equation*}
    \mathcal E
    =
    \left\{
    |\overline{F}^{(m)}_i(t)-\overline F_i(t)|
    \leq
    \frac14\max\{\overline F_i(t),\delta\}
    \text{ for all }i\in[n],\ t\geq0
    \right\}.
\end{equation*}
\end{lemma}

See the proof in \Cref{app:relative-tail-approx}. We will not learn an optimal allocation directly, but rather an optimal vector of added marginal utilities by agent, considering that some baseline proportional to $\delta$ is already being allocated.

We define the shifted empirical best marginal utility for $c \in [0,1/4]$ as
\begin{equation}\label{eq:empirical_curves}
\widehat G_i(c)
    \coloneqq
    \int_0^\infty \min\{(\overline F^{(m)}_i(t)-2\delta)_+,c\}\,dt.
\end{equation}

After observing the training samples, the algorithm chooses a vector of target
activation probabilities. For $\alpha \in (0,1)$, we  define 
\begin{equation}\label{eq:learn-chat-alpha-small}
\widehat c
    \in
    \argmax_{c\in[0,1/4]^n:\sum_i c_i\le\frac14}
    \sum_{i=1}^n \widehat G_i(c_i)^{1-\alpha}.
\end{equation}
For $\alpha=1$, let
\begin{equation}\label{eq:learn-chat-nash}
\widehat c
    \in
    \argmax_{c\in[0,1/4]^n:\sum_i c_i\le\frac14}
    \sum_{\substack{i\in[n] \\
    \sup_{c\in[0,1/4]}\widehat G_i(c)>0}}\log \widehat G_i(c_i), 
\end{equation}
with the convention that if $\sup_{c\in[0,1/4]}\widehat G_i(c)=0$ for all $i \in [n]$, then output $\widehat c=0$.

The learned curves $\widehat G_i$ deliberately ignore the bottom $O(\delta)$
part of each tail. This truncation makes the empirical curves stable under a
relative tail approximation. The small baseline allocation $\delta$ then
ensures that no coordinate is completely lost, which is essential for the
Nash case $\alpha=1$.

We prove in \Cref{app:training-bound} that this learned allocation achieves a constant fraction of the relaxed benchmark when used offline. 
\begin{proposition}\label{prop:training-bound}
There exists a universal constant $C>0$ such that for $m \geq C n\log n$ samples and $\delta=1/(20n)$, the learned target marginal allocation $d=\delta+(\widehat c_1,\dots,\widehat c_n)$  is such that $\sum_{i \in [n]} d_i\leq 1/3$, and satisfies
\begin{equation*}
    W_\alpha\left(
    \be[\mu_1(d_1)],
    \ldots,
    \be[\mu_n(d_n)]
    \right)
    \ge
    \frac{3}{128e}\Gamma_\alpha.
\end{equation*}
\end{proposition}

The intuition is that $\widehat G_i(c)$ estimates the additional utility obtained by increasing agent $i$'s activation probability by $c$, after a small baseline allocation has already been reserved. The appendix formalizes this intuition in two steps: first, empirical marginal curves of the form $\widehat G_i$ approximate the corresponding true marginal gains up to constant factors; second, an offline optimization lemma shows that optimizing these marginal gains is enough, up to another universal constant, to approximate the final shifted $\alpha$-fair welfare. Together, these two ingredients imply that the optimizer $\widehat c$ learns a budget profile whose true relaxed welfare is a constant fraction of $\Gamma_\alpha$.

\subsubsection{From an Offline to Online Allocation}
\label{subsubsec:online-implementation}

It remains to implement the learned target probabilities online. We use fresh samples to convert each target probability $d_i$ into a rank threshold: agent $i$ is eligible if its realized value ranks among the top $k_i\approx (m_{\rm imp}+1)d_i$ values in an independent sample pool. By exchangeability, this implements the top-$d_i$ quantile rule up to a constant factor. Since the total target mass satisfies $\sum_i d_i\le 1/3$, the probability that an agent is blocked by an earlier eligible agent is also bounded by a constant. Thus the online implementation preserves a constant fraction of the learned marginal utilities.

For $\delta>0$ and a vector $d\in[0,1]^n$ satisfying
$d_i\geq\delta$ for every $i$ and $\sum_i d_i \leq1/3$, define $m_{\rm imp}
    =
    \left\lceil \frac{2}{\delta}\right\rceil-1$.
For each agent $i$, draw $m_{\rm imp}$ i.i.d. samples
$\widetilde X_{i,1},\ldots,\widetilde X_{i,m_{\rm imp}}$ from $F_i$, and set $k_i
    =
    \left\lfloor (m_{\rm imp}+1)d_i\right\rfloor$. When $X_i$ arrives, agent $i$ is declared eligible if $X_i$ ranks among the top
$k_i$ values of $\{X_i,\widetilde X_{i,1},\ldots,\widetilde X_{i,m_{\rm imp}}\}$, 
using symmetric tie-breaking. The algorithm then allocates the entire remaining
budget to the first eligible agent and allocates zero to all subsequent agents.
We denote the resulting allocation rule by $\widehat \alloc$. We prove in \Cref{app:rank-threshold-utility} the following conservation of marginal utility through the above online algorithm. 

\begin{proposition}\label{lem:rank-threshold-utility}
For all $i \in [n]$, $\be[\widehat \alloc_i X_i] \geq {\frac{1-e^{-1}}{3}} \mu_i (d_i)$.
\end{proposition}

Putting everything together, we can conclude the proof of the sample competitive ratio result.

\begin{proof}[Proof of \Cref{thm:sample-alpha-small}]
Use $m_{\mathrm{tr}}=Cn\log(n)$ samples to learn $d$, which in particular satisfies the conditions of \Cref{lem:rank-threshold-utility}. The latter gives, for every $i\in[n]$,
\begin{equation*}
u_i(\ALG)=\be_{\mathrm{training}}\left[\be[\widehat\alloc_iX_i\mid d]\right]
\geq\frac{1-e^{-1}}{3}\be[\mu_i(d_i)].
\end{equation*}
Using monotonicity and homogeneity of $W_\alpha$, followed by \Cref{prop:training-bound},
\begin{align*}
W_\alpha(u(\ALG))\geq\frac{1-e^{-1}}{3}W_\alpha(\be[\mu(d)])
\geq\frac{1-e^{-1}}{3}\cdot\frac{3}{128e}\Gamma_\alpha(F)
=\frac{e-1}{128e^2}\Gamma_\alpha(F)
\geq\frac{e-1}{128e^2}\optra,
\end{align*}
where the last inequality uses \Cref{prop:relaxed-prophet-upper}. In total, the algorithm uses $m_{\mathrm{tr}}+m_{\rm imp}=O(n\log n)$ samples per distribution.
\end{proof}

\subsection{No Finite-Sample Algorithms for Constant Competitive Ratio with $\alpha>1$}
\label{subsec:sample-alpha-large}

We finally show that the positive result above cannot be extended beyond $\alpha=1$. 
The positive sample result relies on the fact that, for $\alpha\le1$, losses in a few low-utility coordinates can still be compensated to some extent by higher utilities elsewhere. This fails for $\alpha>1$. In this regime the welfare function behaves like a bottleneck objective: an algorithm must know which coordinate is structurally hard to serve. The following construction makes that coordinate statistically invisible to any finite-sample learner.

\begin{theorem}\label{thm:sample-impossibility-alpha-large}
For every fixed $n\in\bn$, every $m\in\bn$, and every $\alpha>1$, the $m$-sample competitive ratio for the ex-ante $\alpha$-fair welfare objective is $\compra^{n,m}(\alpha)=\frac1n$.
\end{theorem}

\begin{proof} Let $\alpha >1$. 
The lower bound follows from the uniform allocation rule $\alloc_i=1/n$. Indeed, for every instance, this rule achieves utility vector $(\frac{\be[X_1]}{n},\ldots,\frac{\be[X_n]}{n})$, 
while no prophet can give agent $i$ expected utility more than $\be[X_i]$. By monotonicity and homogeneity of $W_\alpha$, this gives a $1/n$ competitive ratio lower bound.

We prove the matching upper bound. We construct $n$ distinct environments. For a fixed $\varepsilon>0$, define the distribution $F$ such that a random variable $Z\sim F$ takes the value $1$ with probability $1-\varepsilon$ and $1/\varepsilon^2$ with probability $\varepsilon$. For each environment $s\in\{1,\ldots,n\}$, let the rewards $(X_1^{(s)},\ldots,X_n^{(s)})$ be independent random variables where
\begin{equation*}
    X_s^{(s)}=1\quad\text{almost surely},
    \qquad
    X_i^{(s)}\sim F\quad\text{for }i\neq s.
\end{equation*}
Let $\Pr_s$ and $\be_s$ denote probability and expectation in environment $s$, over the samples, the online realizations, and the internal randomness of the learning algorithm.

\textbf{Analysis of the prophet benchmark.}
First, we show that in every environment $s$, the prophet obtains a large ex-ante $\alpha$-fair value.  In environment $s$, suppose the prophet allocates to agent $s$ whenever no other agent realizes the high value, and allocates to a high-valued agent otherwise, splitting arbitrarily if several high values occur. Under this policy, the ``bad'' agent $s$ receives expected utility at least $\be_s[\alloc_sX_s^{(s)}]\geq 1-(n-1)\varepsilon$.
For every agent $i\neq s$, conditional on $X_i^{(s)}=1/\varepsilon^2$, no other agent $j\neq s,i$ realizes the high value with probability at least $1-(n-2)\varepsilon$. In that event, the prophet can allocate the full budget to agent $i$. Hence
\begin{equation*}
    \be_s[\alloc_iX_i^{(s)}]
    \geq
    \varepsilon\cdot\frac{1}{\varepsilon^2}\cdot(1-(n-2)\varepsilon)
    =
    (1-O(n\varepsilon))\frac{1}{\varepsilon}.
\end{equation*}
Consequently,
\begin{equation*}
    \optra
    \geq
    W_\alpha(1-(n-1)\varepsilon,(1-O(n\varepsilon))\frac{1}{\varepsilon},\ldots,(1-O(n\varepsilon))\frac{1}{\varepsilon}).
\end{equation*}
Since $\alpha>1$, the terms $(\frac{1}{\varepsilon})^{1-\alpha}$ vanish as $\varepsilon\to0$. Therefore $\optra
    \geq
    n^{\frac{1}{\alpha-1}}(1-o(1))$.

For any learning algorithm $A$, the ex-ante objective in environment $s$ is bounded from above by the expected utility of the bad agent $s$:
\begin{equation*}
    W_\alpha(\be[a \cdot X^{(s)}])
    \leq
    W_\alpha\left(\be_s[\alloc_s],\be_s[\alloc_iX_i^{(s)}]_{i\neq s}\right),
\end{equation*}
and the main point of the proof is to show that, for at least one environment $s$, the learner cannot allocate much budget to this bad agent.

\textbf{Indistinguishability of environments.}
A learner with finitely many samples cannot reliably identify which coordinate is the deterministic bad agent. For each environment $s$, let $\widetilde X_{i,k}^{(s)}$ denote the $k$-th sample from the distribution of $X_i^{(s)}$, for $k\in[m]$. Any two environments $s$ and $s'$ differ only in the distributions of coordinates $s$ and $s'$. Because all rewards and samples are independent, and because the random seed has the same law in all environments, the total variation distance between $\Pr_s$ and $\Pr_{s'}$ is bounded by the sum of the distances between the differing coordinates:
\begin{equation*}
\begin{aligned}
\mathrm{TV}(\Pr_s\Vert\Pr_{s'})
\leq
\mathrm{TV}(X_s^{(s)}\Vert X_s^{(s')})
+
\mathrm{TV}(X_{s'}^{(s)}\Vert X_{s'}^{(s')})
+
\sum_{k=1}^m
\left(
\mathrm{TV}(\widetilde X_{s,k}^{(s)}\Vert \widetilde X_{s,k}^{(s')})
+
\mathrm{TV}(\widetilde X_{s',k}^{(s)}\Vert \widetilde X_{s',k}^{(s')})
\right).
\end{aligned}
\end{equation*}
We compute each individual term. In environment $s$, coordinate $s$ is deterministic equal to $1$, whereas in environment $s'$ it has distribution $F$. Therefore
\begin{align*}
\mathrm{TV}(X_s^{(s)}\Vert X_s^{(s')})
&=
\frac12
\left(
\left|\Pr_s(X_s^{(s)}=1)-\Pr_{s'}(X_s^{(s')}=1)\right|
+
\left|\Pr_s(X_s^{(s)}=1/\varepsilon^2)-\Pr_{s'}(X_s^{(s')}=1/\varepsilon^2)\right|
\right)\\
&=
\frac12(\varepsilon+\varepsilon)
=
\varepsilon.
\end{align*}
The same calculation gives $\mathrm{TV}(X_{s'}^{(s)}\Vert X_{s'}^{(s')})=\varepsilon$.
Likewise, for every sample $k\in[m]$, $\mathrm{TV}(\widetilde X_{s,k}^{(s)}\Vert \widetilde X_{s,k}^{(s')})=\varepsilon$, and $\mathrm{TV}(\widetilde X_{s',k}^{(s)}\Vert \widetilde X_{s',k}^{(s')})=\varepsilon$.
This results in the bound $\mathrm{TV}(\Pr_s\Vert\Pr_{s'})
    \leq
    2(m+1)\varepsilon$.

Since any online allocation $\alloc_i$ is a measurable function of the samples, the observations, the random seed, and takes values in $[0,1]$, the difference in expected allocations across environments is bounded by this TV distance via its variational characterization:
\begin{equation*}
\left|\be_s[\alloc_i]-\be_{s'}[\alloc_i]\right|
\leq
\sup_{\vert f\vert\leq1}
\left|
\be_{Z\sim s}[f(Z)]-\be_{Z\sim s'}[f(Z)]
\right| \notag=
2\mathrm{TV}(\Pr_s\Vert\Pr_{s'})
\leq
4(m+1)\varepsilon.    
\end{equation*}

\textbf{Deriving the $1/n$ competitive ratio.}
We conclude by an averaging argument. In any fixed environment, say $s=1$, the budget constraint requires $\sum_{i=1}^n\be_1[\alloc_i]\leq1$.
Hence there exists some environment index $s\in[n]$ for which $\be_1[\alloc_s]\leq \frac1n$.
By the indistinguishability established above, the learner cannot significantly increase this allocation in environment $s$. Formally, using the indistinguishability  bound,
\begin{equation*}
    \be_s[\alloc_s]
    \leq
    \be_1[\alloc_s]+4(m+1)\varepsilon
    \leq
    \frac1n+4(m+1)\varepsilon.
\end{equation*}
Since $X_s^{(s)}=1$ almost surely, the expected utility of the bad agent $s$ in environment $s$ is at most $\be_s[\alloc_sX_s^{(s)}]
    =
    \be_s[\alloc_s]
    \leq
    \frac1n+4(m+1)\varepsilon$.
Even if we give the learning algorithm utility $\frac{1}{\varepsilon}+1$ for all other agents, monotonicity implies
\begin{equation*}
    W_\alpha(u(A))
    \leq
    W_\alpha\left(\frac1n+4(m+1)\varepsilon,\frac{1}{\varepsilon}+1,\ldots,\frac{1}{\varepsilon}+1\right).
\end{equation*}
Letting $\varepsilon\to0$, the terms $(\frac{1}{\varepsilon}+1)^{1-\alpha}$ vanish. Therefore
\begin{equation*}
    W_\alpha(u(\ALG))
    \leq
    n^{\frac{1}{\alpha-1}-1}(1+o(1)).
\end{equation*}
Combining the prophet lower bound and the online algorithm upper bound, and then sending $\varepsilon\to0$, yields
\begin{equation*}
    \frac{W_\alpha(u(\ALG))}{\optra}
    \leq
    \frac{n^{\frac{1}{\alpha-1}-1}}{n^{\frac{1}{\alpha-1}}}
    =
    \frac1n.
\end{equation*}
Since the learning algorithm was arbitrary, the upper bound follows. Combined with the uniform-allocation lower bound, this proves the claim.
\end{proof}

\begin{remark}
The restriction $\alpha>1$ is essential for this finite-sample impossibility. For $\alpha\in(0,1]$, \Cref{thm:sample-alpha-small} gives a constant competitive ratio with $O(n\log n)$ samples per distribution.
\end{remark}

\section{Ex-Post $\alpha$-Fairness} \label{sec:ex-post-alpha}

We now return to the ex-post interpretation of fairness. For an allocation algorithm $\ALG$, the ex-post $\alpha$-fair objective is defined as
\begin{equation*}
\rp(\ALG)
\coloneqq \be[W_\alpha(a \cdot X)].
\end{equation*}
We denote the optimal online and prophet values by
\begin{equation*}
    \rpon
    \coloneqq
    \sup_{\ALG\in\mathcal A_n}\rp(\ALG),
    \qquad
    \optrp
    \coloneqq
    \sup_{\ALG\in\mathcal L_n}\rp(\ALG),
\end{equation*}
and the corresponding fixed-horizon and all-horizon competitive ratios by
\begin{equation*}
    \comprp^n(\alpha)
    \coloneqq
    \inf_{X_1,\ldots,X_n}
    \frac{\rpon}{\optrp},
    \qquad
    \comprp(\alpha)
    \coloneqq
    \inf_{n\geq1}\comprp^n(\alpha).
\end{equation*}

\subsection{A Universal Competitive Ratio Lower Bound for $\alpha \in (0,1)$}

We start by giving a full-information algorithm that yields a universal constant over the whole interval $(0,1)$.

Throughout the sections for $\alpha \in (0,1)$, we will work for ease of computation with a re-parametrized version of the problem. More specifically, we let
\begin{equation*}
q=\frac{\alpha}{1-\alpha}, \quad  Y_i=X_i^{1/q}, \quad R=\sum_{i=1}^n Y_i.
\end{equation*}
Equivalently, $\alpha=q/(q+1)$. Dropping the normalization factor, which does not affect the competitive ratio, the alpha-fair payoff can therefore be expressed as, 
\begin{equation*}
n^{q+1} W_\alpha(a \cdot X)=\Psiq(a,Y)\coloneqq \left( \sum_{i =1}^n  a_i^{1/(q+1)} Y_i^{q/(q+1)}\right)^{q+1}
\end{equation*}
With this expression in terms of the $Y_i$'s, the optimal offline value becomes clear. 
\begin{lemma}\label{lem:alpha_opt}
For $y \in \mathbb{R}^n_+$,  $\sup_{a_i\ge0,\,\sum_i a_i\le1}\Psiq(a,y)=(\sum_{i=1}^n y_i)^q$, and is achieved at $a_i=y_i/(\sum_{j=1}^n y_j)$.
\end{lemma}
\begin{proof}
Hölder's inequality gives
\begin{equation*}
\sum_i a_i^{1/(q+1)}y_i^{q/(q+1)}
\le
\big(\sum_i a_i\big)^{1/(q+1)}
\big(\sum_i y_i\big)^{q/(q+1)} 
\le
\big(\sum_i y_i\big)^{q/(q+1)}
\end{equation*}
using that for $a$ feasible, $\sum_{i=1}^n a_i \leq 1$. Raising to the power $q+1$ yields the upper bound $(\sum_{i=1}^n y_i)^q$. If all the $y_i=0$, then both sides vanish. Suppose now that at least one $y_i >0$. The allocation $a_i=y_i/(\sum_{j=1}^n y_j)$ achieves the upper bound and is feasible.
\end{proof}

\begin{theorem}\label{thm:cr_uniform}
For every $\alpha\in(0,1)$, \Cref{alg:posterior-hazard} achieves $\comprp(\alpha)\geq e^{-\pi^2/6}\approx 0.193$.
\end{theorem}

We next introduce the biased probability distribution $\Pr^{(q)}$, with Radon--Nikodym derivative
\begin{equation*}
\frac{d\Pr^{(q)}}{d\Pr}=\frac{R^q}{\be[R^q]},
\qquad
\be^{(q)}[Z]=\frac{\be[ZR^q]}{\be[R^q]}.
\end{equation*}
The purpose of this biased law, is to evaluate in some ways the performance of the algorithm relative to the prophet performance on the same realization. In fact, $\be[\mathrm{ALG}]/\be[\mathrm{OPT}]=\be^{(q)}[\mathrm{ALG}/\mathrm{OPT}]$. So we will essentially  prove a lower bound on the expectation of the ratio under the biased law, which will transfer to a lower bound on the ratio of expectations for the original law.

We also define the Kullback-Leibler divergence for two discrete distributions $a$ and $a'$ supported on $[n]$ as 
\begin{equation*}
\KL(a^*\Vert a)=\sum_{i=1}^n a_i^*\log\frac{a_i^*}{a_i},
\end{equation*}
where a term with $a_i^*=0$ is interpreted as $0$, and the divergence is infinite if $a_i^*>0$ but $a_i=0$. If $a$ does not spend the entire budget, we append its unused budget as an extra coordinate, with corresponding prophet allocation $0$.

We now prove a generic upper bound that relates the KL divergence between the online allocation and the optimal offline one. The idea is, now that we are working with the biased law, we want to compare the online algorithm performance to the prophet performance realization by realization. To achieve this, we essentially want to show that the online and offline allocation policies are close to one another. Given that we can interpret fractional allocations as (sub)-probability distributions on $[n]$, an appropriate measure of distance could be found in statistics, such as $f$-divergence. In this specific problem, the distance, or the affinity, that naturally emerges is the $\alpha$-Rényi Affinity. This affinity is related to the KL affinity $\exp(-\KL)$ as the first converges to the second as $\alpha \to 1$, and we can in fact lower bound the Rényi Affinity by the KL one throughout $\alpha \in (0,1)$. This is what explains the possibility of a universal bound that does not depend on $\alpha$. 

The following lemma relates the loss in welfare to the expected KL divergence from the prophet allocation. Its proof only uses the weighted AM--GM inequality.

\begin{lemma}\label{lem:KL_bound}
Let $a$ be a feasible online allocation. Then, for every $\alpha\in(0,1)$,
\begin{equation*}
\frac{\be[W_\alpha(a\cdot X)]}{\optrp}\geq \exp\big(-\be^{(q)}[\KL(a^*\Vert a)]\big).
\end{equation*}
\end{lemma}

\begin{proof}
By \Cref{lem:alpha_opt}, the prophet allocation is $a_i^*=Y_i/R$, and the normalization factor $n^{-(q+1)}$ cancels in the competitive ratio. Thus
\begin{equation*}
\frac{\be[W_\alpha(a\cdot X)]}{\optrp}
=\be^{(q)}\bigg[\frac{\Psiq(a,Y)}{R^q}\bigg]
=\be^{(q)}\left[\left(\sum_{i=1}^n a_i^*\left(\frac{a_i}{a_i^*}\right)^{1/(q+1)}\right)^{q+1}\right].
\end{equation*}
Using the weighted AM--GM inequality and $\sum_i a_i^*=1$ gives
\begin{equation*}
\left(\sum_{i=1}^n a_i^*\left(\frac{a_i}{a_i^*}\right)^{1/(q+1)}\right)^{q+1}
\geq \prod_{i=1}^n\left(\frac{a_i}{a_i^*}\right)^{a_i^*}
=\exp\bigg( \sum_{i=1}^n a_i^* \log( \frac{a_i}{a_i^*})\bigg)=\exp\big(-\KL(a^*\Vert a)\big).
\end{equation*}
Taking expectation under $\Pr^{(q)}$ and applying Jensen's inequality proves the claim.
\end{proof}

Hence, the above lemma suggests that if we want to obtain an algorithm with a good competitive ratio, we must choose the online allocation $a$ so that $\KL(a^*\Vert a)$ is bounded. A necessary condition for it to be true is for $a^*$ to be absolutely continuous with respect to $a$, or in other words, so that $a_i^*>0$ implies $a_i>0$. To satisfy this condition, we make it so that the online decision maker never runs out of budget. We adapt a similar idea to the single-sample algorithm, and draw a fresh sample $\widetilde R$ to mimic  the prophet performance, and allocate something close to $Y_i/(\widetilde R)$ times the remaining budget. We instead use $\be^{(q)}[Y_i/(Y_i+\widetilde R) \mid Y_i]$, with the expectation to obtain a deterministic rule, and the different normalization to make sure that this allocation remains bounded by $1$, ensuring that unless $\widetilde R=0$ we will not run out of budget. We define in \Cref{alg:posterior-hazard} the specific allocation rule. 

\begin{algorithm}[t]
\caption{Biased-Total Algorithm for $\alpha\in(0,1)$}
\label{alg:posterior-hazard}
\KwIn{The distributions of $Y_1,\ldots,Y_n$ and $q=\alpha/(1-\alpha)$.}
Let $h(y)=\be^{(q)}[y/(y+\widetilde R)]$, where $\widetilde R$ is an independent total drawn from the biased law\;
$b\gets1$\;
\For{$i=1,\ldots,n$}{
Observe $Y_i$\;
$a_i\gets b\,h(Y_i)$\;
$b\gets b-a_i$\;
}
\end{algorithm}

We next remark that the optimal allocation admits a probabilistic interpretation through exponential random variables. Conditionally on $Y$, let $E_1,\ldots,E_n$ be independent exponential random variables with  respective rates $Y_1,\ldots,Y_n$, where $E_i=\infty$ if $Y_i=0$, and let $J=\arg\min_i E_i$. Then, by the exponential race property, the probability that $E_i$ is the minimum is $Y_i/\sum_j Y_j$, so 
\[
\Pr^{(q)}(J=i\mid Y)=\Pr(J=i\mid Y)=\frac{Y_i}{R}=a_i^*.
\]
The first equality is due to the fact that once $Y$ is fixed, the exponential random variables are not affected by the bias. 
Thus, the optimal allocation is precisely the conditional distribution of the first exponential clock to ring. 
We define for all $t \geq 0$, the survival function of the minimum $E_i$ as 
\begin{equation*}
F(t)\coloneqq \Pr^{(q)}(\min_{i \in [n]} E_i \geq t).
\end{equation*}
We prove in \Cref{app:exponential-race} a lower bound on the biased probability that $i$ is the smallest exponential variable given $E_i$.

\begin{lemma}\label{lem:exponential-race}
We have $\Pr^{(q)}(J=i\mid E_i)\geq F(E_i)$. 
\end{lemma}

We next use \Cref{lem:exponential-race} to bound the expected KL divergence between the prophet's allocation and our allocation under the biased law. We define an auxiliary allocation based on $F$, and show that its conditional expectation is exactly the allocation $\alloc_i$ from \Cref{alg:posterior-hazard}.  

\begin{proposition}\label{prop:hazard_CR}
The allocation rule in \Cref{alg:posterior-hazard} satisfies $\be^{(q)}[\KL(a^*\|a)]\leq\frac{\pi^2}{6}$.
\end{proposition}

\begin{proof}
Consider the auxiliary allocation $\widehat a_i=F(E_i)\prod_{j<i}(1-F(E_j))$. It allocates a fraction $F(E_i)$ of the remaining budget, which is a lower bound on agent $i$'s winning probability conditional on its $E_i$. We can express $F(t)$ as an expectation using independence: 
\begin{equation*}
F(t)=\Pr^{(q)}(\min_i E_i\geq t)=\be^{(q)}[\Pr^{(q)}(E_i\geq t\text{ for every }i\mid Y)]=\be^{(q)}[\prod_i\Pr^{(q)}(E_i\geq t\mid Y)]
=\be^{(q)}[e^{-tR}].
\end{equation*}
Hence, for $y>0$, we can show that the conditional expected target auxiliary allocation exactly corresponds to our target allocation. 
\begin{align*}
\be^{(q)}[F(E_i)\mid Y_i=y]=\int_0^\infty ye^{-yt}F(t)\,dt=\int_0^\infty ye^{-yt}\be^{(q)}[e^{-tR}]\,dt&=\be^{(q)}\left[\int_0^\infty ye^{-yt}e^{-tR}\,dt\right]\\
&=\be^{(q)}\left[\int_0^\infty ye^{-t(y+R)}\,dt\right]\\
&=\be^{(q)}\left[\frac{y}{y+R}\right]=h(y),
\end{align*}
and the identity also holds at $y=0$. Conditional independence of the exponential random variables then gives
\[
\be^{(q)}[\widehat a_i\mid Y]
=h(Y_i)\prod_{j<i}(1-h(Y_j))
=a_i.
\]
Thus, the algorithm is the conditional average of the auxiliary allocation. Since $\Pr^{(q)}(J=i\mid Y)=a_i^*$,
\begin{equation*}
a_i
=\be^{(q)}[\widehat a_i\mid Y]
\geq\be^{(q)}[\widehat a_i\mathbf{1}_{\{J=i\}}\mid Y]
=a_i^*\be^{(q)}[\widehat a_i\mid Y,J=i].
\end{equation*}
Then, we can lower bound the expected log auxiliary allocation by Jensen's inequality
\begin{equation*}
\be^{(q)}[-\log\widehat a_i\mid Y,J=i]
\geq-\log\be^{(q)}[\widehat a_i\mid Y,J=i]
\geq-\log\left(\frac{a_i}{a_i^*}\right)
=\log\left(\frac{a_i^*}{a_i}\right).
\end{equation*}
Multiplying by $a_i^*$ and summing the previous inequality, we can upper bound the $\KL$:
\begin{equation*}
\KL(a^*\|a)
=\sum_{i}a_i^*\log\left(\frac{a_i^*}{a_i}\right)\leq\sum_{i}\Pr^{(q)}(J=i\mid Y)
\be^{(q)}[-\log\widehat a_i\mid Y,J=i]
=\be^{(q)}[-\log\widehat a_J\mid Y].
\end{equation*}

In addition, we have that $F(E_J)$ is uniform, as $F$ is the survival function of $E_J$ under the biased law.

By the definition of $\widehat a$,
\begin{align*}
-\log\widehat a_J
=-\log\big(F(E_J)\prod_{i<J}(1-F(E_i))\big)
&=-\log F(E_J)+\sum_{i<J}\log\frac{1}{1-F(E_i)}\\
&\leq-\log F(E_J)+\sum_{i \in [n]} \ind{i \neq J}\log\frac{1}{1-F(E_i)}.
\end{align*}

By \Cref{lem:exponential-race}, $\Pr^{(q)}(J=i\mid E_i)\geq F(E_i)$. Consequently, on the event $F(E_i)>0$,
\begin{align*}
\Pr^{(q)}(J\neq i\mid E_i)
=1-\Pr^{(q)}(J=i\mid E_i)
&\leq\frac{\Pr^{(q)}(J=i\mid E_i)}{F(E_i)}
-\Pr^{(q)}(J=i\mid E_i)\\
&=\Pr^{(q)}(J=i\mid E_i)\frac{1-F(E_i)}{F(E_i)}.
\end{align*}
When $F(E_i)=0$, $E_i$ is infinite, so $J\neq i$ and the corresponding loss $\log(1/(1-F(E_i)))$ is zero. 
For each $i$, conditioning on $E_i$, applying the preceding inequality, and then using the tower property again gives
\begin{align*}
\be^{(q)}\left[\mathbf{1}_{\{J\neq i\}}\log\frac{1}{1-F(E_i)}\right]&=\be^{(q)}\left[
\be^{(q)}\left[\mathbf{1}_{\{J\neq i\}}\log\frac{1}{1-F(E_i)}
\,\middle|\,E_i\right]\right]\\
&=\be^{(q)}\left[
\be^{(q)}[\mathbf{1}_{\{J\neq i\}}\mid E_i]
\log\frac{1}{1-F(E_i)}\right]\\
&=\be^{(q)}\left[
\Pr^{(q)}(J\neq i\mid E_i)\log\frac{1}{1-F(E_i)}\right]\\
&\leq\be^{(q)}\left[
\Pr^{(q)}(J=i\mid E_i)
\frac{1-F(E_i)}{F(E_i)}\log\frac{1}{1-F(E_i)}\right]\\
&=\be^{(q)}\left[
\be^{(q)}[\mathbf{1}_{\{J=i\}}\mid E_i]
\frac{1-F(E_i)}{F(E_i)}\log\frac{1}{1-F(E_i)}\right]\\
&=\be^{(q)}\left[
\mathbf{1}_{\{J=i\}}\frac{1-F(E_i)}{F(E_i)}
\log\frac{1}{1-F(E_i)}\right].
\end{align*}
Here the factors involving $F(E_i)$ can be taken outside conditional expectations given $E_i$, since they are functions of $E_i$.

Taking expectations in the initial bound and summing these inequalities, we obtain
\begin{align*}
\be^{(q)}[-\log\widehat a_J]
&\leq\be^{(q)}[-\log F(E_J)]
+\sum_i\be^{(q)}\left[
\mathbf{1}_{\{J\neq i\}}\log\frac{1}{1-F(E_i)}\right]\\
&\leq\be^{(q)}[-\log F(E_J)]
+\sum_i\be^{(q)}\left[
\mathbf{1}_{\{J=i\}}\frac{1-F(E_i)}{F(E_i)}
\log\frac{1}{1-F(E_i)}\right]\\
&=\be^{(q)}\left[
-\log F(E_J)
+\frac{1-F(E_J)}{F(E_J)}
\log\frac{1}{1-F(E_J)}\right].
\end{align*}
The last equality holds because exactly one index satisfies $J=i$, so the sum retains only the term evaluated at $E_J$. Because $F$ is the survival function of $E_J$, $F(E_J)$ is a uniform random variable on $[0,1]$ under the biased law. Putting everything together,
\begin{equation*}
\be^{(q)}[\KL(a^*\Vert a)]
\leq\int_0^1\left(-\log u+\frac{1-u}{u}\log\frac{1}{1-u}\right)\,du.
\end{equation*}
The integral can be evaluated precisely. Write $(1-u)/u=1/u-1$ and use by symmetry $\int_0^1-\log u\,du=\int_0^1-\log(1-u)\,du=1$. This gives
\begin{align*}
\int_0^1\left(-\log u+\frac{1-u}{u}\log\frac{1}{1-u}\right)\,du&=1+\int_0^1\frac{-\log(1-u)}{u}\,du-1
=\int_0^1\sum_{k=1}^{\infty}\frac{u^{k-1}}{k}\,du\\
&=\sum_{k=1}^{\infty}\frac1k\int_0^1u^{k-1}\,du=\sum_{k=1}^{\infty}\frac{1}{k^2}
=\frac{\pi^2}{6}. \qedhere
\end{align*}
\end{proof}

\subsection{A Single Sample Competitive Ratio Lower Bound For $\alpha \in (0,1)$}\label{sec:one-sample}

In this section, we prove that  some sample prophet inequality guarantee is possible, with only having access to a single sample per distribution.

Throughout, we suppose that $\alpha$, and therefore $q$ is known. Let $\widetilde X_i$ be one independent sample following the same law as $X_i$, or equivalently let $\widetilde Y_i$ be one independent sample with the same law as $Y_i$. We suppose that the $\widetilde Y_i$ are observed before having to allocate the resource online for the $Y_i$. We define $\widetilde R= \sum_{i=1}^n \widetilde Y_i$. Thus $\widetilde R$ is an independent copy of $R$. 

\begin{algorithm}[t]
\SetAlgoNoLine
\KwIn{Independent samples $\widetilde X_1,\dots,\widetilde X_n$ with $\widetilde X_i \sim X_i$, online observations $X_1,\ldots,X_n$.}
\KwOut{A feasible online allocation $a_1,\ldots,a_n$.}

Draw $\xi\sim \mathrm{Ber}(1/2)$;

\eIf{$\xi=1$}{
Let $T=\max_{j\in[n]}\widetilde X_j$;
Allocate the entire budget to the first item $i$ such that $X_i\ge T$
}{
Let $\widetilde R=\sum_{j=1}^n \widetilde X_j^{(1-\alpha)/\alpha}$

\eIf{$\widetilde R=0$}{
Allocate nothing
}{
Allocate $\displaystyle a_i=\min \{ \frac{X_i^{(1-\alpha)/\alpha}}{\widetilde R}, \text{remaining budget}\}$, until the budget is exhausted.
}
}
\caption{Single-Sample Algorithm for $(0,1)$-Ex-Post Fairness}
\label{alg:single-sample-spike-fill}
\end{algorithm}

The ex-post prophet value for $\alpha<1$ is governed by the aggregate
$R=\sum_iY_i$. This aggregate can be large for two different reasons: either
one coordinate is very large, or many coordinates contribute moderate mass.
The single-sample algorithm handles these two cases separately by mixing a
spike-capturing threshold rule with a filling rule based on the sampled
aggregate, and is defined in \Cref{alg:single-sample-spike-fill}. 

The first algorithm is the same as the single sample algorithm in the classical utilitarian setting, and simply allocates all the mass to the first item $Y_i \geq \max_{j \in [n]} \widetilde Y_j$, its goal is to capture realizations when most of the value is concentrated on one coordinate, that is to say when $\max_{i \in [n]} Y_i$ is large. It is known from \cite{rubinstein2020optimal} that this algorithm achieves at least $\frac{1}{2}\be[\max_{i \in [n]} X_i]=\frac{1}{2}\be[\max_{i \in [n]} Y_i^q] $ in expectation. 

The second algorithm, which we will call the Sample-Target-Filling algorithm, stems from the following observation: if $R$ were known in advance, we could do the optimal allocation in an online fashion. Of course, $R$ is not known, even if we had access to the full distributions, so instead we simulate it through $\widetilde R$: the second algorithm thus treats $\widetilde R$ as the true realization of $R$, and allocates $a_i=Y_i/\widetilde R$ until either the sequence finishes or the online decision maker runs out of budget first. The goal of this algorithm is to handle realizations where the value is spread out over multiple coordinates.

\begin{theorem} \label{thm:single_sample_CR}
We have for all $\alpha \in (0,1)$, that $\comprp(\alpha) \geq \frac{1}{4 \cdot 3^{\frac{\alpha}{1-\alpha}}}$, and this is achieved by the single-sample \Cref{alg:single-sample-spike-fill}. 
\end{theorem}

We first give a guarantee for this second algorithm. 

\begin{lemma}\label{lem:STF_perf}
The Sample-Target-Filling algorithm achieves at least $\frac{1}{2}\be[\min\{ R^q,\widetilde R^q\}]$ in expectation. 
\end{lemma}

\begin{proof}
If $\widetilde R=0$, allocate nothing. Suppose that $ \widetilde R \leq  R$. Then the algorithm runs out of budget, meaning that $\sum_{i =1}^n a_i=1$, and $a_i  \leq Y_i/\widetilde R$ so that $Y_i \geq a_i \widetilde R$. Hence
\begin{equation*}
\left( \sum_{i=1}^n a_i^{\frac{1}{q+1}} Y_i^{\frac{q}{q+1}} \right)^{q+1}\geq\left( \sum_{i=1}^n a_i \widetilde R^{\frac{q}{q+1}}\right)^{q+1}=\widetilde R^q \left( \sum_{i=1}^n a_i \right)^{q+1}=\widetilde R^{q}.
\end{equation*}
Now, suppose artificially that we were to run the same algorithm, but taking $Y_i$ as the samples and $\widetilde Y_i$ as the values relevant for the allocation. Then, if $R \leq \widetilde R$ the algorithm achieves at least $ R^{q}$. As a consequence, the performance of the sum of these two processes, using $\widetilde Y_i$ for samples and $Y_i$ for the allocation and vice versa, is at least $\min\{ R^q, \widetilde R^q\}$. Hence, in expectation, the sum achieves $\be[\min\{ R^q, \widetilde R^q\}]$. By exchangeability and symmetry, the expected performance of these algorithms is the same, and therefore if the expected sum of those two processes is at least $\be[\min\{ R^q, \widetilde R^q\}]$, then the expected performance of the Sample-Target-Filling algorithm is $\frac{1}{2}\be[\min\{ R^q, \widetilde R^q\}]$. 
\end{proof}

We now prove a moment decomposition upper bound on the prophet performance. See the proof in \Cref{app:moment-decomp}.

\begin{lemma}\label{lem:moment-decomp}
For every $q >0$,  $\be[R^q]
\le
3^q\left(
\be[\max_{i \in [n]}Y_i^q]
+
\be[\min\{R,R'\}^q]
\right)$.
\end{lemma}

\begin{proof}[Proof of \Cref{thm:single_sample_CR}]
The algorithm runs either the max-sample or Sample-Target-Filling with equal probability $1/2$. Hence by \cite{rubinstein2020optimal} and \Cref{lem:STF_perf},
\begin{equation*}
\rp(A)=
\frac12\be[\textnormal{ALG-Max-Sample}]
+
\frac12\be[\textnormal{ALG-Target-Filling}] \ge
\frac14\left(
\be[\max_{i \in [n]} Y_i^q]
+
\be[\min\{R,R'\}^q]
\right).
\end{equation*}
Hence, by \Cref{lem:moment-decomp,lem:alpha_opt}, $\rp(A) \geq \frac{1}{4} \frac{1}{3^q}
\be[R^q] =  \frac{1}{4} \frac{1}{3^q}  \optrp$. \qedhere
\end{proof}

\subsection{Uniform Sample Impossibility for $\alpha \in (0,1)$}

The two previous results, \Cref{thm:single_sample_CR,thm:cr_uniform}, might lead us to believe that, as in the ex-ante setting, a sample budget depending only on $n$ suffices for a universal constant competitive ratio. We show that this is not possible. In fact, for every fixed horizon and every finite sample budget, the best guarantee approaches the trivial $1/n$ bound as $\alpha$ approaches $1$ from below.

Write $\comprp^{n,m}(\alpha)$ for the best worst-case competitive ratio among algorithms using at most $m$ samples per marginal. The hard instances below have a hidden cutoff, and the positive values increase geometrically before this cutoff. The learner cannot identify the cutoff from samples, while the prophet allocation concentrates on its last positive coordinate.

\begin{theorem}\label{thm:sample-impossibility-stated}
For every $n\geq2$ and every finite $m\in\bn\cup\{0\}$, $\displaystyle \lim_{\alpha\to 1^-}\comprp^{n,m}(\alpha)=\frac1n$. 
\end{theorem}

In other words, for every finite sample budget $m(n)$ depending only on $n$, there is a sequence $\alpha_n\in(0,1)$ tending to $1$ such that $\comprp^{n,m(n)}(\alpha_n)\leq(1+o(1))/n$.

This result does not contradict \Cref{thm:single_sample_CR}, which gives a positive ratio for each fixed $\alpha<1$. The endpoint $\alpha=1$ is an exceptional case, as uniform allocation is pointwise optimal there. 

\begin{proof}
Let $n\geq2$, finite $m\geq0$, $\varepsilon\in(0,1)$, and $q>0$. We construct $n$ distinct environments. For all $s\in[n]$, consider the environment with independent values
\begin{equation*}
X_i=
\begin{cases}
\varepsilon^{(1-i)q},&\text{with probability }\varepsilon,\\
0,&\text{otherwise},
\end{cases}
\quad \text{for }i\leq s,
\qquad \mbox{and }
X_i=0\quad \text{for }i>s.
\end{equation*}

Fix any learning algorithm. We first consider its behavior in environment $n$, conditional on all training values being zero and all $n$ online values being positive. By feasibility, the expected allocations (taking expectation with respect to the algorithm internal randomness) along this path sum to at most $ \sum_i \alloc_i \leq 1$, so some index $s\in[n]$ receives expected allocation at most $1/n$. We thus choose environment $s$ as the hard instance.
Now, on environment $s$, condition on all training values being zero and on $X_1,\ldots,X_s>0$. In this case, the algorithm sees exactly the same history through time $s$ as in the realization we considered for environment $n$. Its conditional expected allocation at time $s$ is therefore at most $1/n$. With all expectations in the following display taken under this conditioning, concavity and Jensen's inequality gives
\begin{align*}
\be[\Psiq(\alloc,Y)]\leq \Psiq(\be[\alloc],Y)
=\bigg(\sum_{i=1}^s\varepsilon^{(1-i)\frac{q}{q+1}}\be[\alloc_i]^{\frac{1}{q+1}}\bigg)^{q+1}\leq\bigg(\sum_{i=1}^{s-1}\varepsilon^{(1-i)\frac{q}{q+1}}+n^{-\frac{1}{q+1}}\varepsilon^{(1-s)\frac{q}{q+1}}\bigg)^{q+1},
\end{align*}
where the second inequality uses $\be[\alloc_s]\leq1/n$ and $\alloc_i\leq1$.

On the other hand, if any of the first $s$ online values is zero, then
\begin{equation*}
R=\sum_{i=1}^s X_i^{1/q}
\leq\sum_{i=1}^s\varepsilon^{1-i}-\min_{1\leq i\leq s}\varepsilon^{1-i}
=\sum_{i=1}^s\varepsilon^{1-i}-1.
\end{equation*}
Finally, the probability of a nonzero training value is at most $mn\varepsilon$, independently of the online values, and on this event, we simply bound the algorithm by the prophet value. Combining these bounds gives
\begin{align*}
\be[\Psiq(\alloc,Y)]\leq\varepsilon^s\bigg(\sum_{i=1}^{s-1}\varepsilon^{(1-i)\frac{q}{q+1}}+n^{-\frac{1}{q+1}}\varepsilon^{(1-s)\frac{q}{q+1}}\bigg)^{q+1}\quad+\bigg(\sum_{i=1}^s\varepsilon^{1-i}-1\bigg)^q+mn\varepsilon\,\be[R^q].
\end{align*}
Since the first $s$ online values are all positive with probability $\varepsilon^s$, the normalized prophet value satisfies $\be[R^q ]\geq\varepsilon^s(\sum_{i=1}^s\varepsilon^{1-i})^q$. Putting everything together,
\begin{align*}
\frac{\be[\Psiq(\alloc,Y)]}{\be[R^q]}
&\leq\frac{\big(\sum_{i=1}^{s-1}\varepsilon^{(1-i)\frac{q}{q+1}}+n^{-\frac{1}{q+1}}\varepsilon^{(1-s)\frac{q}{q+1}}\big)^{q+1}}{\big(\sum_{i=1}^s\varepsilon^{1-i}\big)^q}+\varepsilon^{-s}\bigg(\frac{\sum_{i=1}^s\varepsilon^{1-i}-1}{\sum_{i=1}^s\varepsilon^{1-i}}\bigg)^q+mn\varepsilon\\
&=\frac{\big(\sum_{i=1}^{s-1}\varepsilon^{i\frac{q}{q+1}} +n^{-\frac{1}{q+1}}\big )^{q+1}}{\big(\sum_{i=0}^{s-1}\varepsilon^i\big)^q}+\varepsilon^{-s}\bigg(1-\frac{1}{\sum_{i=1}^s\varepsilon^{1-i}}\bigg)^q+mn\varepsilon.
\end{align*}
The last equality follows by factoring out $\varepsilon^{(1-s)q}$ from the numerator and denominator of the first fraction and reindexing the sums. For fixed $\varepsilon$ and $s$, the second term tends to zero as $q\to\infty$ (equivalently, $\alpha \to 1$), while the first tends to
\begin{equation*}
\left(\sum_{i=0}^{s-1}\varepsilon^i\right)\exp\left(-\frac{\log n+(\log\varepsilon)\sum_{i=1}^{s-1}i\varepsilon^i}{\sum_{i=0}^{s-1}\varepsilon^i}\right).
\end{equation*}
This expression tends to $1/n$ as $\varepsilon\to0$. Since there are only finitely many possible cutoffs, these estimates are uniform over algorithms. Taking $q\to\infty$ and then $\varepsilon\to 0$ therefore yields
\begin{equation*}
\limsup_{\alpha\to 1^-}\comprp^{n,m}(\alpha)\leq\frac1n.
\end{equation*}

For the matching lower bound, uniform allocation yields a performance of at least $1/n$ of the prophet, as $W_\alpha(x/n)
\geq W_\alpha(\alloc\cdot x/n)
=\frac1n W_\alpha(\alloc \cdot x)$. 
\end{proof}

\subsection{Full Information Impossibility Result for $\alpha >1$}\label{sec:expost-alpha}

The universal guarantee of \Cref{thm:cr_uniform} might suggest that the competitive ratio degrades gradually as $\alpha$ goes over $1$ in the full information setting. Instead, we show a sharp phase transition: for every $\alpha>1$ and every horizon $n$, the competitive ratio is exactly $1/n$. The lower bound is given by uniform allocation. To prove the matching upper bound, we use independent values on a geometric scale and let the ratio between consecutive scales tend to infinity.

We first record the prophet benchmark for the ex-post $\alpha$-fair objective when $\alpha>1$. Its proof is given in \Cref{app:alpha_prophet}.

\begin{lemma}\label{lem:alpha_prophet}
For $\alpha>1$ and $x\in\mathbb R_{++}^n$, $\sup_{a_i\geq0,\,\sum_i a_i\leq1}W_\alpha(a\cdot x)
=n^{-1/(1-\alpha)}(\sum_{i=1}^n x_i^{\frac{1-\alpha}{\alpha}})^{\frac{\alpha}{1-\alpha}}$,
and the supremum is achieved at $a_i=x_i^{\frac{1-\alpha}{\alpha}}/(\sum_{j=1}^n x_j^{\frac{1-\alpha}{\alpha}})$. If some $x_i=0$, every feasible allocation has welfare zero.
\end{lemma}

\begin{proposition}\label{thm:alpha-expost-hardness-upper}
For every $\alpha>1$ and every $n\geq2$, $\comprp^n(\alpha)\leq1/n$.
\end{proposition}

\begin{proof}
Fix $\alpha>1$, $n\geq2$, and $b>1$. Let $M>1$. Set $X_1=1$ deterministically, and for $i\in\{2,\ldots,n\}$ let the values be independent with
\begin{equation*}
X_i=\begin{cases}
b^{-(i-1)},&\text{with probability }1-1/b,\\
M,&\text{with probability }1/b.
\end{cases}
\end{equation*}
Let $I$ be the last index taking its low value, with $I=1$ if every $X_i$, $i\geq2$, equals $M$.

\noindent \textbf{Upper bound on the online value.} 
Let $\alloc$ be any online allocation. 
For $i\geq2$, the event $\{I=i\}$ means that $X_i$ takes its low value and every subsequent value equals $M$. Hence,
\begin{align*}
\be[a_i\ind{I=i}]=\be\left[a_i\ind{X_i=b^{-(i-1)}}
\ind{X_{i+1}=\cdots=X_n=M}\right]
&=\be[a_i\ind{X_i=b^{-(i-1)}}]\,
\Pr(X_{i+1}=\cdots=X_n=M)\\
&=\be[a_i\ind{X_i=b^{-(i-1)}}]\,
\prod_{k=i+1}^n\Pr(X_k=M)\\
&=b^{-(n-i)}\be[a_i\ind{X_i=b^{-(i-1)}}].
\end{align*}
The second equality uses that $\alloc_i$ is adapted to the filtration, and therefore depends only on $X_1$ to $X_i$, so is independent of the remaining random variables. Similarly, $\be[a_1\ind{I=1}]
=\be[a_1]\Pr(X_2=\cdots=X_n=M)
=b^{-(n-1)}\be[a_1]$.

We now partition the expectation according to the possible values of $I$. On $\{I=i\}$, we have $\min_j a_jX_j\leq a_iX_i=a_i b^{-(i-1)}$. Therefore,
\begin{align*}
&\be\big[\min_{j\in[n]}a_jX_j\big]=\sum_{i=1}^n\be\big[\min_{j\in[n]}a_jX_j\,\ind{I=i}\big]\leq\sum_{i=1}^n b^{-(i-1)}\be[a_i\ind{I=i}]\\
&=b^{-(n-1)}\be[a_1]
+\sum_{i=2}^n b^{-(i-1)}b^{-(n-i)}
\be[a_i\ind{X_i=b^{-(i-1)}}]\\
&=b^{-(n-1)}
\big(\be[a_1]+\sum_{i=2}^n\be[a_i\ind{X_i=b^{-(i-1)}}]\big)\\
&\leq b^{-(n-1)}
\big(\be[a_1]+\sum_{i=2}^n\be[a_i]\big)=b^{-(n-1)}\be\big[\sum_{i=1}^n a_i\big]\leq b^{-(n-1)}.
\end{align*}
The last inequality uses that the allocation is feasible, and thus $\sum_i \alloc_i \leq 1$. This already gives the upper bound for $\alpha \to \infty$, to make it work for $\alpha >1$, we use the following inequality. For every nonnegative vector $z$, because $1/(1-\alpha)<0$,
\begin{equation*}
W_\alpha(z)=\big(\frac1n\sum_{i=1}^n z_i^{1-\alpha}\big)^{1/(1-\alpha)}
\leq\big(\frac1n(\min_i z_i)^{1-\alpha}\big)^{1/(1-\alpha)}
=n^{1/(\alpha-1)}\min_i z_i.
\end{equation*}
It follows that $\rp(\ALG)\leq n^{1/(\alpha-1)}b^{-(n-1)}$ for every online algorithm, including randomized ones. \medskip

\noindent \textbf{Lower bound on the prophet.} On $\{I=i\}$, every $X_k$ with $k>i$ equals $M$. For $k\leq i$, we always have $X_k\geq b^{-(k-1)}$, since $X_k$ is either its low value or $M>1$. Because $-(\alpha-1)/\alpha<0$, this implies
\begin{equation*}
X_k^{-(\alpha-1)/\alpha}
\leq \left(b^{-(k-1)}\right)^{-(\alpha-1)/\alpha}
=b^{(\alpha-1)(k-1)/\alpha}.
\end{equation*}
Splitting the sum at $i$, we therefore obtain
\begin{align*}
\sum_{k=1}^n X_k^{-(\alpha-1)/\alpha}=\sum_{k=1}^i X_k^{-(\alpha-1)/\alpha}
+\sum_{k=i+1}^n M^{-(\alpha-1)/\alpha}&=\sum_{k=1}^i X_k^{-(\alpha-1)/\alpha}
+(n-i)M^{-(\alpha-1)/\alpha}\\
&\leq\sum_{k=1}^i b^{(\alpha-1)(k-1)/\alpha}
+(n-i)M^{-(\alpha-1)/\alpha}\\
&\leq\sum_{k=1}^i b^{(\alpha-1)(k-1)/\alpha}
+nM^{-(\alpha-1)/\alpha}.
\end{align*}
For the remaining geometric sum, factor out its largest term and reverse the order of summation:
\begin{align*}
\sum_{k=1}^i b^{(\alpha-1)(k-1)/\alpha}
=b^{(\alpha-1)(i-1)/\alpha}
\sum_{k=1}^i b^{-(\alpha-1)(i-k)/\alpha}
&=b^{(\alpha-1)(i-1)/\alpha}
\sum_{\ell=0}^{i-1}\left(b^{-(\alpha-1)/\alpha}\right)^\ell\\
&=b^{(\alpha-1)(i-1)/\alpha}
\frac{1-b^{-(\alpha-1)i/\alpha}}{1-b^{-(\alpha-1)/\alpha}}\\
&\leq\frac{b^{(\alpha-1)(i-1)/\alpha}}{1-b^{-(\alpha-1)/\alpha}},
\end{align*}
Hence,  $\sum_{k=1}^n X_k^{-(\alpha-1)/\alpha}
\leq\frac{b^{(\alpha-1)(i-1)/\alpha}}{1-b^{-(\alpha-1)/\alpha}}+nM^{-(\alpha-1)/\alpha}$. 
Thus, for fixed $b$ and $n$, \Cref{lem:alpha_prophet} implies that the pointwise prophet value on this event is at least $n^{1/(\alpha-1)}(1-b^{-(\alpha-1)/\alpha})^{\alpha/(\alpha-1)}b^{-(i-1)}-o_M(1)$.
By \Cref{lem:alpha_prophet}, conditioning on the possible values of $I$ and applying the pointwise lower bound gives
\begin{align*}
\optrp
&=n^{1/(\alpha-1)}
\sum_{i=1}^n\Pr(I=i)\,
\be\bigg[
\big(\sum_{k=1}^n X_k^{-(\alpha-1)/\alpha}\big)^{-\alpha/(\alpha-1)}
\,|\,I=i\bigg]\\
&\geq n^{1/(\alpha-1)}(1-b^{-(\alpha-1)/\alpha})^{\alpha/(\alpha-1)}
\sum_{i=1}^n\Pr(I=i)b^{-(i-1)}-o_M(1).
\end{align*}
For fixed $\alpha$, $b$, and $n$, the error term can be chosen uniformly over the finitely many indices $i$. Its probability-weighted sum is therefore still $o_M(1)$.
Using $\Pr(I=1)=b^{-(n-1)}$ and $\Pr(I=i)=(1-1/b)b^{-(n-i)}$ for $i\geq2$, the remaining sum becomes
\begin{align*}
\sum_{i=1}^n\Pr(I=i)b^{-(i-1)}=\Pr(I=1)+\sum_{i=2}^n\Pr(I=i)b^{-(i-1)}
&=b^{-(n-1)}
+\sum_{i=2}^n(1-1/b)b^{-(n-i)}b^{-(i-1)}\\
&=b^{-(n-1)}
+(1-1/b)\sum_{i=2}^n b^{-(n-1)}\\
&=b^{-(n-1)}
+(n-1)(1-1/b)b^{-(n-1)}\\
&=b^{-(n-1)}\left(1+(n-1)(1-1/b)\right).
\end{align*}
Substituting this in the prophet lower bound yields
\begin{equation*}
\optrp\geq n^{1/(\alpha-1)}(1-b^{-(\alpha-1)/\alpha})^{\alpha/(\alpha-1)}b^{-(n-1)}
\left(1+(n-1)(1-1/b)\right)-o_M(1).
\end{equation*}\medskip

Combining the bounds and first taking $M\to\infty$, we obtain
\begin{equation*}
\comprp^n(\alpha)\leq\frac{1}{(1-b^{-(\alpha-1)/\alpha})^{\alpha/(\alpha-1)}\left(1+(n-1)(1-1/b)\right)}.
\end{equation*}
Finally, taking $b\to\infty$ gives $\comprp^n(\alpha)\leq1/n$.
\end{proof}

The matching lower bound follows directly from monotonicity and homogeneity.
\begin{theorem}\label{thm:alpha-expost-hardness}
For every $\alpha>1$ and every $n\geq1$, $\comprp^n(\alpha)=1/n$. In particular, $\comprp(\alpha)=0$.
\end{theorem}

\begin{proof}
The uniform allocation $a_i=1/n$ satisfies, for every feasible prophet allocation $a^*$ and every realized vector $x$,
\begin{equation*}
W_\alpha(x/n)\geq W_\alpha(a^*\cdot x/n)=\frac1n W_\alpha(a^*\cdot x).
\end{equation*}
Taking expectations proves $\comprp^n(\alpha)\geq1/n$. The upper bound is \Cref{thm:alpha-expost-hardness-upper} for $n\geq2$, and the case $n=1$ is immediate.
\end{proof}

\section{Conclusion and Future Directions}

We study prophet inequalities for $\alpha$-fair welfare objectives, distinguishing between ex-ante and ex-post notions of fairness. Our results show that these two interpretations lead to sharply different algorithmic phenomena. In the ex-ante model, the classical $\nicefrac12$ prophet inequality extends to the full $\alpha$-fair family under full distributional knowledge, but the sample-access model exhibits a phase transition at $\alpha=1$. In the ex-post model, constant guarantees are possible for $\alpha<1$, while for $\alpha>1$ the competitive ratio is exactly $1/n$ even with full information. Thus, the interaction between fairness, online decision-making, and statistical information depends delicately on both the value of $\alpha$ and the level at which fairness is enforced.

Our work leaves several interesting open questions. 
For example, in the ex-ante sample-access model, our algorithm uses $O(n\log n)$ samples per distribution for $\alpha\in(0,1]$. It would be interesting to determine whether a constant number of samples per distribution suffices, as in the classical utilitarian setting, or whether a growing number of samples is necessary.
Additionally, for the ex-post objective with $\alpha\in(0,1)$, our full-information results give a universal constant lower bound, but we do not know the optimal constant. The natural conjecture is that the sharp competitive ratio is $\nicefrac12$, matching the classical prophet inequality, at least throughout the range $\alpha\in(0,1)$. A sharpened specialized proof of the $\KL$ argument in fact already yields a competitive ratio of $1/2$, but is specialized to the case $\alpha=1/2$. 
Finally, our ex-ante full-information result is geometric and therefore immediately extends to an XOS multi-resource setting. Understanding whether analogous guarantees hold for ex-post $\alpha$-fairness in such combinatorial settings remains open.

Beyond these more directly related open questions, our work opens the door to a broader theory of prophet inequalities for non-linear welfare objectives.  For instance, on the ex-ante side, our geometric argument applies to every symmetric, concave, homogeneous welfare function. On the ex-post side, our hardness for $\alpha>1$ shows that such generality is impossible, even for independent values. Identifying the structural properties of nonlinear objectives that permit constant ex-post prophet inequalities remains a central open direction.
\section*{Acknowledgments}

This project has been partially funded by the European Research Council (ERC) under the European Union's Horizon Europe Program (FACT, grant agreement No.~101170373), by an Amazon Research Award, and by the Israel Science Foundation Breakthrough Program (grant No.~2600/24).

\newpage

\appendix

\section{AI and LLM usage}
We disclose the use of Artificial Intelligence (AI) and Large Language Models (LLMs) in the preparation of this manuscript as follows:
\begin{itemize}
\item[(1)] Gemini~3.1 and ChatGPT~5.5 were used for editorial assistance, including polishing, rephrasing, and improving the clarity of parts of the introduction and related work. No citations or BibTeX entries were generated using generative AI.
\item[(2)] Gemini~3.1 and ChatGPT~5.5 were used for high-level brainstorming about possible proof approaches for some results. The final proof strategies, formal arguments, theorem statements, calculations, and verification are entirely the authors' own, and all mathematical content was checked by hand.
\end{itemize}

\section{Further Related Work} \label{app:related-work}

\paragraph{Fair online algorithms.}
Fairness in online allocation, often motivated by repeated distribution of scarce resources, has been studied both through explicit constraints on admissible policies and through alternative welfare objectives.
 Closest in motivation, but distinct in setting, is concurrent and independent work by \citet{Yang2026}, which studies generalized-mean/CES welfare objectives in online allocation. Their benchmark and guarantees differ from ours: they prove regret guarantees for repeated online allocation, whereas we prove prophet-inequality guarantees for $\alpha$-fair objectives in stochastic online selection and allocation. \citet{En2026} study online allocation with individual Lipschitz fairness constraints, requiring similar agents to receive similar expected assignment values; their work maximizes linear welfare subject to these constraints and proves regret guarantees against a fair fluid benchmark.

In the canonical independent-values prophet model, our \emph{ex-ante} Rawlsian objective, corresponding to the limit $\alpha \to \infty$, coincides with the max--min stopping formulation of \citet{HillKennedy1990}, who prove a tight $\nicefrac{1}{2}$ competitive ratio. We extend their guarantee to the whole class of $\alpha$-fair welfare functions.

For the \emph{ex-post} objective, the closest analogue is the minimum fill-rate model of \citet{Lien2014nonprofit,Yuanzheng2023,Sinclair2020SequentialFA,Sinclair2022sequentialfair} in the case $\alpha \to \infty$, where (possibly correlated) stochastic demands $D_i$ arrive online and performance is measured by fill rate $\min\{\alloc_i/D_i,1\}$. Although distinct, this objective aligns closely with ours under the change of variables $D_i=1/X_i$: the truncation corresponds to a data-dependent cap $\alloc_i\le 1/X_i$, which is typically nonbinding when $X_i$ is large.

Recent competitive analyses for this model include \citet{manshadi2023fair}, who give tight guarantees against a weaker deterministic benchmark (harmonic mean of expectations rather than expectation of the harmonic mean), and \citet{Sankararaman}, who recover the tight $\nicefrac{1}{2}$ guarantee of \citet{HillKennedy1990} in the independent ex-ante setting, and show for the ex-post objective that when $D_i \leq O(\epsilon^2/\log(1/\epsilon))$, the competitive ratio is greater than $1-\epsilon$. Both of these studies raise the question of whether a lower bound (positive result) better than $O(1/\log(n))$ for the ex-post competitive ratio under independent demand is possible. The hard instance of \Cref{thm:alpha-expost-hardness} solves this open question by the negative.

More broadly, fairness in prophet inequalities has also been studied by restricting admissible policies (e.g., identity-agnostic or group constraints) rather than by changing the objective \citep{Correa2025fair,Arsenis2022}. Related min--max notions appear in online learning via regret formulations \citep{Balseiro2025}. Finally, \citet{Hajiaghayi2022santa,Kawase} study a multi-resource max--min objective under adversarial values and random order in a converse online model (items arrive online rather than agents).

\paragraph{Prophet inequalities with samples.}
A growing body of work studies prophet settings in which the underlying distributions are unknown and can only be accessed via samples.
This learning variant was initiated by \citet{Azar2014}, who showed that in many settings a constant competitive ratio is achievable via samples, through a reduction to order-oblivious secretary algorithms. \citet{rubinstein2020optimal} then proved a striking result: a \emph{single} sample per distribution already attains the optimal $\nicefrac{1}{2}$ ratio for utilitarian welfare. \citet{Correa2019unknown,Correa2024sample,rubinstein2020optimal,Guo2021} explore the tradeoff between samples and achievable competitive ratio in the i.i.d. setting, implying that a constant number of samples (order $n$ in total) is sufficient to get an approximately optimal competitive ratio. Single sample results for the prophet-secretary model appear in \cite{Correa2020googol}.
\citet{Cristi2024prophet} show that a constant number of samples also suffices in the prophet-secretary and the free-order models.
 
Additional single-sample constant-competitive ratio results are available for combinatorial settings with unit-demand, budget-additive, and XOS valuations \cite{CaramanisDFFLLP22,dutting2024online}. In sharp contrast to this series of positive results, our separation theorem shows that for the ex-ante model for $\alpha >1$ fairness, worst-case instances admit no improvement over the trivial $1/n$ guarantee with any finite number of samples.

\section{Deferred Proofs of \Cref{sec:exante-alpha}}

\subsection{Proof of \Cref{prop:alpha-general-upper}} \label{app:alpha-general-upper} 

We first state a lemma regarding SCH functions. 

\begin{lemma}\label{lem:sch-properties}
Every $\mathrm{SCH}$ function is continuous on $\br_{++}^n$ and coordinatewise monotone on $\br_+^n$.
\end{lemma}

\begin{proof}
Continuity on $\br_{++}^n$ follows from the standard fact that any finite concave function is continuous on the relative interior of its domain.
It remains to prove monotonicity. Let $x,z\in\br_+^n$. By concavity and homogeneity,  $f(x+z)
    =2f\left(\frac{x+z}{2}\right)
    \geq f(x)+f(z)$ implies super-additivity. Hence, for  $x\leq y$ coordinatewise, and $z=y-x\in\br_+^n$, $f(y)=f(x+z)\geq f(x)+f(z)\geq f(x)$ which proves the claim by setting all coordinates of $z$ to $0$ except at each index $i\in [n]$ successively.
\end{proof}

\begin{proof}[Proof of \Cref{prop:alpha-general-upper}]
We first note that $f(1,1)>0$. Since $f$ is nonzero, there exists $x\in\br_+^2$ such that $f(x)>0$. Let $M>0$ be such that $(M,M)\geq x$ coordinatewise. By monotonicity, which follows from non-negativity, concavity, and homogeneity, we have $f(M,M)\geq f(x)>0$. By homogeneity, $f(M,M)=Mf(1,1)$, and therefore $f(1,1)>0$.

Consider a two-agent instance parameterized by $\varepsilon>0$:
\begin{equation*}
    X_1=1,
    \qquad
    X_2=
    \begin{cases}
    1/\varepsilon, & \text{with probability } \varepsilon,\\
    0, & \text{otherwise.}
    \end{cases}
\end{equation*}
Let $\ALG$ be an online algorithm, and let $p=\be[\alloc_1]$. Since $X_1=1$, we have $u_1(\ALG)=p$. Moreover, the budget constraint gives $\alloc_2\leq 1-\alloc_1$ almost surely. Since $\alloc_1$ is chosen before $X_2$ is observed and $X_2$ is independent of the algorithm's internal randomness,
\begin{equation*}
    u_2(\ALG)
    =
    \be[\alloc_2X_2]
    \leq
    \be[(1-\alloc_1)X_2]
    =
    (1-p)\be[X_2]
    =
    1-p.
\end{equation*}
Thus every online expected-utility vector is dominated by $(p,1-p)$ for some $p\in[0,1]$. 

By monotonicity, $f(u_1(\ALG),u_2(\ALG))
    \leq
    f(p,1-p)$. 
By symmetry, $f(p,1-p)=f(1-p,p)$.
By concavity,
\begin{equation*}
    f\left(\frac12,\frac12\right)
    =
    f\left(\frac{(p,1-p)+(1-p,p)}{2}\right)
    \geq
    \frac12 f(p,1-p)+\frac12 f(1-p,p)
    =
    f(p,1-p).
\end{equation*}
By homogeneity, $f\left(\frac12,\frac12\right)
    =
    \frac12 f(1,1)$.
Therefore, the online value is at most $\frac12 f(1,1)$.

The prophet can achieve the diagonal vector $(1-\varepsilon,1-\varepsilon)$: allocate the full budget to agent $1$ when $X_2=0$, and allocate a fraction $1-\varepsilon$ to agent $2$ when $X_2=1/\varepsilon$. Hence the prophet value is at least $ f(1-\varepsilon,1-\varepsilon)
    =
    (1-\varepsilon)f(1,1)$,
using homogeneity. Therefore, for this instance, the competitive ratio is at most $\frac{\frac12 f(1,1)}{(1-\varepsilon)f(1,1)}
    =
    \frac{1}{2(1-\varepsilon)}$.
Taking $\varepsilon\to0$ gives the desired upper bound.
\end{proof}

\subsection{Proof of \Cref{lem:relative-tail-approx}} \label{app:relative-tail-approx}

\begin{proof}
Fix an agent $i$. The class of upper intervals $\mathcal R= \{(t,\infty):t\ge0\}$
 has VC dimension one. By the relative $(p,\varepsilon)$-approximation theorem for finite-VC range spaces, applied with $p=\delta$, $\varepsilon=1/4$, and failure probability $\eta=1/(4n)$, a sample of size $C\delta^{-1}\bigl(\log(1/\delta)+\log n\bigr)$
is, with probability at least $1-\eta$, a relative $(\delta,1/4)$-approximation for the range space generated by $\mathcal R$; see \citet[Definition~2.3 and Theorem~2.11(ii)]{Peled2011}\footnote{They state their results for finite range spaces with normalized counting measure, so the same statement applies to i.i.d. samples from an arbitrary probability measure by viewing the sample as the empirical measure.}. 

Thus, simultaneously for every $t\ge0$, the empirical tail probability $\overline F_i^{(m)}(t)$ satisfies $\vert \overline F_i^{(m)}(t)-\overline F_i(t)|
    \le
    \frac14\overline F_i(t)$ whenever $\overline F_i(t)\ge\delta$ and $|\overline F_i^{(m)}(t)-\overline F_i(t)|
    \le
    \frac14\delta$ whenever $\overline F_i(t)\le\delta$. Equivalently $|\overline F_i^{(m)}(t)-\overline F_i(t)|
    \le
    \frac14\max\{\overline F_i(t),\delta\}$ for all $t \geq 0$. 
Taking a union bound over the $n$ agents gives failure probability at most $n\eta=1/4$. Hence $\Pr(\mathcal E)\ge3/4$.
\end{proof}

\subsection{Proof of \Cref{prop:training-bound}} \label{app:training-bound}

Before proving the desired result, we start by proving two auxiliary lemmas. 

The first one tells us that this is a good approximation of the utility gained by allocating an additional budget $c$ to an agent. 
\begin{lemma}\label{lem:sandwich}
Let $m \geq 1$. The function $\widehat G_i$ is nondecreasing and concave. Moreover, on $\mathcal E$, for every $i\in[n]$ and every $c\in[0,1/4]$,
\begin{equation*}
    \frac12\bigl(\mu_i(4\delta+c)-\mu_i(4\delta)\bigr)
    \leq
   \widehat G_i(c)
    \leq
    2\bigl(\mu_i(\delta+c)-\mu_i(\delta)\bigr).
\end{equation*}
\end{lemma}

\begin{proof}
The first claim follows because, for each $t$, the map $c\mapsto \min\{(\overline F^{(m)}_i(t)-2\delta)_+,c\}$ is nondecreasing and concave.

Fix $i$ and $t$, and write $s=\overline F_i(t)$ and $\widehat s=\overline F^{(m)}_i(t)$. On $\mathcal E$, $|\widehat s-s|\leq \frac14\max\{s,\delta\}$. For the upper bound, if $s<\delta$, then $(\widehat s-2\delta)_+=0$. If $s\geq\delta$, then $(\widehat s-2\delta)_+
    \leq
    \frac54s-2\delta
    \leq
    2(s-\delta)_+$.
Thus $\min\{(\widehat s-2\delta)_+,c\}\leq2\min\{(s-\delta)_+,c\}$.

For the lower bound, if $s<4\delta$ the right-hand side is zero. If $s\geq4\delta$, then
\begin{equation*}
    (\widehat s-2\delta)_+
    \geq
    \frac34s-2\delta
    \geq
    \frac12(s-4\delta)_+.
\end{equation*}
Hence $\min\{(\widehat s-2\delta)_+,c\}\geq\frac12\min\{(s-4\delta)_+,c\}$. Integrating the two pointwise bounds and using the identity
\begin{equation*}
    \mu_i(b+c)-\mu_i(b)
    =
    \int_0^\infty
    \left(
    \min\{\overline F_i(t),b+c\}
    -
    \min\{\overline F_i(t),b\}
    \right)dt 
    =
    \int_0^\infty
    \min\{(\overline F_i(t)-b)_+,c\}\,dt, 
\end{equation*}
completes the proof.
\end{proof}

The second ingredient is an offline optimization lemma showing that, for $\alpha\leq 1$, optimizing the marginal utility vector alone gives a constant-factor approximation even after adding an arbitrary baseline utility vector.

\begin{lemma}\label{lem:optimizer-bound}
Let $D\subseteq\br_+^n$ be convex compact. Fix $f\in\br_+^n$ and $r\in D$.
\begin{itemize}
    \item \emph{$\alpha \in (0,1)$:} If $h$ maximizes $\sum_i y_i^{1-\alpha}$
over $D$, then $W_\alpha(f+h)
    \ge
    e^{-1}W_\alpha(f+r)$.
    \item \emph{$\alpha=1$:} For $h$ maximizing $\sum_{i \in[n] : \sup_{y\in D}y_i>0}\log y_i$ over $D$, then $W_1(f+h)
    \ge
    e^{-1}W_1(f+r)$.
\end{itemize}
\end{lemma}

\begin{proof}
First suppose $\alpha\in(0,1)$. Coordinates that are identically zero on $D$ have $h_i=r_i=0$ and contribute the same term to both sides. Since adding a common nonnegative term to both power sums preserves the desired inequality, it suffices to consider coordinates that are not identically zero on $D$. On those coordinates, optimality gives $h_i>0$.

Since $r\in D$ and $D$ is convex, $(1-\lambda)h+\lambda r\in D$ for every
$\lambda\in[0,1]$. The optimality of $h$ therefore implies that
\begin{equation*}
   (1-\alpha)\sum_i h_i^{-\alpha}(r_i-h_i)= \frac{d}{d\lambda}
    \sum_i\bigl((1-\lambda)h_i+\lambda r_i\bigr)^{1-\alpha}
    \big|_{\lambda=0}
    \le 0.
\end{equation*}
Since $1-\alpha>0$, this is equivalent to $\sum_i h_i^{-\alpha}r_i
    \le
    \sum_i h_i^{1-\alpha}$.
Let $H=\sum_i h_i^{1-\alpha}$, $\pi_i=h_i^{1-\alpha}/H$, $S_i=f_i/h_i$, and $T_i=r_i/h_i$. Then $\sum_i\pi_i=1$ and by the previous optimality inequality,
\begin{equation*}
    \sum_i\pi_iT_i
    =
    \frac{\sum_i h_i^{-\alpha}r_i}{\sum_i h_i^{1-\alpha}}
    \le
    1.
\end{equation*}
Concavity of $x\mapsto x^{1-\alpha}$ gives $x^{1-\alpha}
    \le
    1+(1-\alpha)(x-1)$.
Applying this to $x=(S_i+T_i)/(S_i+1)$ yields
\begin{align*}
    (S_i+T_i)^{1-\alpha}
    \le
    (S_i+1)^{1-\alpha}
    +(1-\alpha)(S_i+1)^{-\alpha}(T_i-1) 
    \le
    (S_i+1)^{1-\alpha}
    +(1-\alpha)T_i.
\end{align*}
Averaging with weights $\pi_i$,
\begin{align*}
    \sum_i\pi_i(S_i+T_i)^{1-\alpha}
    \le
    \sum_i\pi_i(S_i+1)^{1-\alpha}
    +(1-\alpha)\sum_i\pi_iT_i \le
    \sum_i\pi_i(S_i+1)^{1-\alpha}
    +(1-\alpha).
\end{align*}
Since $(S_i+1)^{1-\alpha}\ge1$ and $\sum_i\pi_i=1$, we have $\sum_i\pi_i(S_i+1)^{1-\alpha}\ge1$.
Therefore
\begin{align*}
    \sum_i\pi_i(S_i+1)^{1-\alpha}
    +(1-\alpha) \le
    \left(1+(1-\alpha)\right)
    \sum_i\pi_i(S_i+1)^{1-\alpha} \le
    e^{1-\alpha}
    \sum_i\pi_i(S_i+1)^{1-\alpha}.
\end{align*}
Consequently,
\begin{equation*}
    \sum_i\pi_i(S_i+T_i)^{1-\alpha}
    \le
    e^{1-\alpha}
    \sum_i\pi_i(S_i+1)^{1-\alpha}.
\end{equation*}
Multiplying by $H$ gives  $\sum_i(f_i+r_i)^{1-\alpha}
    \le
    e^{1-\alpha}
    \sum_i(f_i+h_i)^{1-\alpha}$.
Taking the power $1/(1-\alpha)$ proves $W_\alpha(f+r)
    \le
    e W_\alpha(f+h)$, which is the desired inequality.

Now suppose $\alpha=1$. If $I=\{i:\sup_{y\in D}y_i>0\}=\emptyset$, then $r=0$ by the definition of $I$, so the claim is immediate. Assume $I\neq\emptyset$. If $W_1(f+r)=0$, the claim is trivial. Otherwise $f_i+r_i>0$ for every $i$. For $i\notin I$, we have $h_i=r_i=0$, and these coordinates contribute the same term to $W_1(f+h)$ and $W_1(f+r)$.
On $I$, the first-order condition for the logarithmic maximizer gives 
\begin{equation*}
    0
     \ge \frac{d}{d\lambda}\sum_{i\in I}
    \log\bigl(h_i+\lambda(r_i-h_i)\bigr) \big \vert_{\lambda=0}
    =
    \sum_{i\in I}\frac{r_i-h_i}{h_i} 
    =
    \sum_{i\in I}\frac{r_i}{h_i}-|I|,
\end{equation*}
hence $\sum_{i\in I}\frac{r_i}{h_i}
    \le
    |I|$.
Therefore,
\begin{align*}
    \log\frac{W_1(f+r)}{W_1(f+h)}
    =
    \frac1n
    \sum_{i\in I}
    \log\frac{f_i+r_i}{f_i+h_i} \le
    \frac1n
    \sum_{i\in I}
    \log\left(1+\frac{r_i}{h_i}\right) \le
    \frac1n
    \sum_{i\in I}\frac{r_i}{h_i} \le
    \frac{|I|}{n}
    \le
    1.
\end{align*}
Exponentiating gives $W_1(f+r)\le eW_1(f+h)$.
\end{proof}

We can now prove the result.

\begin{proof}[Proof of \Cref{prop:training-bound}]
First, the learned vector satisfies
\begin{equation*}
\sum_{i \in [n]} d_i= n\delta +\sum_{i \in [n]} \widehat c_i \leq \frac{1}{20} +\frac{1}{4} \leq \frac{1}{3}. 
\end{equation*}

It suffices to prove the bound deterministically on the event $\mathcal E$, and then lower bound $\Pr(\mathcal E)$ and the concavity of $W_\alpha$. By \Cref{lem:relative-tail-approx}, we have for $\delta=1/(20n)$ that $\Pr(\mathcal{E})\geq 3/4$.

Let $b^*$ optimize $\Gamma_\alpha(F)$ and set $\overline b_i=b_i^*/4$. Since $\mu_i$ is concave and $\mu_i(0)=0$,
\begin{equation*}
    \mu_i(\overline b_i)
    =
    \mu_i\left(\frac{b_i^*}{4}\right)
    =
    \mu_i\left(\frac14 b_i^*+\frac34\cdot0\right)
    \ge
    \frac14\mu_i(b_i^*)+\frac34\mu_i(0)
    =
    \frac14\mu_i(b_i^*).
\end{equation*}
Thus, coordinatewise, $\mu(\overline b)
    \ge
    \frac14\mu(b^*)$.
By monotonicity and homogeneity of $W_\alpha$,
\begin{align*}
    W_\alpha(\mu(\overline b))
    \ge
    W_\alpha\left(\frac14\mu(b^*)\right) 
    =
    \frac14 W_\alpha(\mu(b^*)) 
    =
    \frac14\Gamma_\alpha(F).
\end{align*}

Set $f_i=\mu_i(\delta)$ and $r_i=\frac12(\mu_i(\overline b_i)-\mu_i(4\delta))_+$. Since $\mu_i$ is concave and $\mu_i(0)=0$, we have
\begin{equation*}
    \mu_i(\delta)
    =
    \mu_i\left(\frac14(4\delta)+\frac34\cdot0\right)
    \ge
    \frac14\mu_i(4\delta),
\end{equation*}
so $\mu_i(4\delta)\le4\mu_i(\delta)=4f_i$.
We now prove that $\mu_i(\overline b_i)\le4(f_i+r_i)$.
If $\mu_i(\overline b_i)\le \mu_i(4\delta)$, then
\begin{equation*}
    \mu_i(\overline b_i)
    \le
    \mu_i(4\delta)
    \le
    4f_i
    \le
    4(f_i+r_i).
\end{equation*}
If $\mu_i(\overline b_i)>\mu_i(4\delta)$, then by definition of $r_i$,
\begin{align*}
    \mu_i(\overline b_i)
    =
    \mu_i(4\delta)+2r_i 
    \le
    4f_i+2r_i 
    \le
    4(f_i+r_i).
\end{align*}
Hence by homogeneity and monotonicity
\begin{equation*}
    W_\alpha(f+r) \geq W_\alpha(\frac{1}{4}\mu(\overline b))
    =
    \frac14 W_\alpha(\mu(\overline b))
    \geq
    \frac1{16}\Gamma_\alpha(F).
\end{equation*}

Let $c_i^*=(\overline b_i-4\delta)_+$, $h_i=\widehat G_i(\widehat c_i)$, and let $D= \{ y \in \br_+^n : \exists c \in [0,1/4]^n, \sum_{i \in [n]} c_i \leq 1/4, 0\leq y_i \leq \hat{G}_i(c_i) \}$. The set $D$ is compact convex, and since the objectives in the definition of $\widehat c$ are coordinatewise increasing, the vector $h$ maximizes the corresponding objective over $D$. Then $\sum_i c_i^*  \leq \sum_i b_i^*/4 \leq1/4$, and by the lower sandwich bound in \Cref{lem:sandwich}, $r_i\leq\widehat G_i(c_i^*)$. Thus $r\in D$. By \Cref{lem:optimizer-bound},
\begin{equation*}
    W_\alpha(f+h)
    \geq
    e^{-1}W_\alpha(f+r).
\end{equation*}
By the upper bound in \Cref{lem:sandwich},
\begin{equation*}
    h_i=
    \widehat G_i(\widehat c_i)
    \leq
    2(\mu_i(d_i)-\mu_i(\delta))
    =2(\mu_i(d_i)-f_i),
\end{equation*}
so $\mu_i(d_i)\geq f_i+h_i/2\geq(f_i+h_i)/2$. Therefore, on $\mathcal E$, 
\begin{equation*}
W_\alpha(\mu(d))
    \geq
    \frac12W_\alpha(f+h). 
\end{equation*}
Taking expectations over the training samples, combining everything, and using $\Pr(\mathcal E)\geq3/4$ gives
\begin{equation*}
    \be_{\mathrm{tr}}[W_\alpha(\mu(d))]
    \geq\frac{1}{16}  \cdot \frac{1}{e} \cdot \frac{1}{2} \cdot \frac{3}{4} \cdot  \Gamma_\alpha(F)=
    \frac{3}{128e}\Gamma_\alpha(F).
\end{equation*} 
Finally, concavity of $W_\alpha$ gives  $W_\alpha(\be[\mu(d)])
    \geq
    \be[W_\alpha(\mu(d))]$ which completes the proof.
\end{proof}

\subsection{Proof of \Cref{lem:rank-threshold-utility}} \label{app:rank-threshold-utility}

We will use the following approximate Jensen inequality for binomial random variables.

\begin{lemma}\label{lem:binomial-tail}
For every $N\geq1$, $k\in[N]$, and $r\in[0,1]$, if $Z\sim\mathrm{Bin}(N,r)$, then
\begin{equation*}
\be[\min\{Z,k\}]\geq(1-e^{-1})\min\{\be[Z],k\}.
\end{equation*}
\end{lemma}

\begin{proof}
For every integer $z\geq0$, $k(1-(1-1/k)^z) \leq k$ and $k(1-(1-1/k)^z)=\sum_{j=0}^{z-1} (1-1/k)^j\leq z$, hence $\min\{z,k\}\geq k(1-(1-1/k)^z)$ with the convention $0^0=1$. Therefore, using the formula for the probability generating function of a binomial random variable,
\begin{equation*}
\be[\min\{Z,k\}]
\geq  k\bigg(1-\be\left[\left(1-\frac{1}{k}\right)^Z\right]\bigg) =k\bigg(1-\bigg(1-\frac rk\bigg)^N\bigg)
\geq k(1-e^{-Nr/k})
\geq(1-e^{-1})\min\{Nr,k\}.
\end{equation*}
The last inequality uses $1-e^{-t}\geq(1-e^{-1})\min\{t,1\}$ for $t\geq0$, which follows from concavity on $[0,1]$ and monotonicity on $[1,\infty)$.
\end{proof}

\begin{proof}[Proof of \Cref{lem:rank-threshold-utility}]
Since $m_{\rm imp}+1\ge 2/\delta$ and $d_i\ge\delta$, we have $ (m_{\rm imp}+1)d_i\ge2$.
Hence $\left\lfloor (m_{\rm imp}+1)d_i\right\rfloor
    \ge
    \frac12(m_{\rm imp}+1)d_i$.
Therefore,  $\frac{k_i}{m_{\rm imp}+1}
    =
    \left\lfloor (m_{\rm imp}+1)d_i\right\rfloor/(m_{\rm imp}+1)
    \ge
    d_i/2$. On the other hand,  $k_i/(m_{\rm imp}+1)=\lfloor (m_{\rm imp}+1)d_i\rfloor/(m_{\rm imp}+1)
    \le
    d_i$.

First ignore the possibility that an earlier agent passes the rank test. For every $t\ge0$, let $Z_t^{(i)}
    =
    \ind{X_i>t}
    +
    \sum_{\ell=1}^{m_{\rm imp}}
    \ind{\widetilde X_{i,\ell}>t} \sim  \mathrm{Bin}(m_{\rm imp}+1,\overline F_i(t))$.
By exchangeability, conditional on the multiset $\{X_i,\widetilde X_{i,1},\ldots,\widetilde X_{i,m_{\rm imp}}\}$, the online value $X_i$ is equally likely to be any one of the
$m_{\rm imp}+1$ values. Among the $Z_t^{(i)}$ values above $t$, exactly
$\min\{Z_t^{(i)},k_i\}$ pass the rank test. Therefore
\begin{equation*}
    \Pr\left(
    X_i>t\text{ and agent }i\text{ passes the rank test}
    \,\middle|\,
    \{X_i,\widetilde X_{i,1},\ldots,\widetilde X_{i,m_{\rm imp}}\}
    \right)
    =
    \frac{\min\{Z_t^{(i)},k_i\}}{m_{\rm imp}+1}.
\end{equation*}
Taking expectations gives
\begin{equation*}
    \Pr(X_i>t\text{ and agent }i\text{ passes the rank test})
    =
    \be\left[
    \frac{\min\{Z_t^{(i)},k_i\}}{m_{\rm imp}+1}
    \right].
\end{equation*}
By Fubini, the previous inequality, and \Cref{lem:binomial-tail},
\begin{align*}
    \be\left[X_i \ind{\text{agent }i\text{ passes the rank test}}\right]&=\be\left[
    \int_0^\infty
    \ind{X_i>t}
    \ind{\text{agent }i\text{ passes the rank test}}
    \,dt
\right]\\
    &=
    \int_0^\infty
    \Pr(X_i>t\text{ and agent }i\text{ passes the rank test})\,dt \\
    & = \int_0^\infty \be\left[
    \frac{\min\{Z_t^{(i)},k_i\}}{m_{\rm imp}+1}
    \right] \,dt \\ 
    &\ge
    {(1-e^{-1})}
    \int_0^\infty
    \min\left\{
    \overline F_i(t),
    \frac{k_i}{m_{\rm imp}+1}
    \right\}\,dt \\
    &=
    {(1-e^{-1})}
    \mu_i\left(\frac{k_i}{m_{\rm imp}+1}\right) \\
    &\ge
    {\frac{1-e^{-1}}{2}}\mu_i(d_i),
\end{align*}
where the last inequality follows from concavity of $\mu_i$, $\mu_i(0)=0$, and
$k_i/(m_{\rm imp}+1)\ge d_i/2$.

The rank-test events are independent across agents, since they depend on independent online values and independent implementation samples. Moreover, agent $j$ passes the rank test with probability $\frac{k_j}{m_{\rm imp}+1}
    \le
    d_j$.
Thus
\begin{equation*}
    \Pr(\text{some }j<i\text{ passes the rank test})
    \le
    \sum_{j<i}d_j
    <
    \frac13.
\end{equation*}
Therefore agent $i$ is preceded by no earlier agent passing the rank test with probability at least $2/3$. This event is independent of whether agent $i$ passes the rank test. Hence
\begin{equation*}
    \be[\widehat \alloc_i X_i]
    \ge
    {\frac{1-e^{-1}}{2}}  \cdot \frac23\mu_i(d_i) 
    =
    {\frac{1-e^{-1}}{3}}\mu_i(d_i). \qedhere
\end{equation*}
\end{proof}

\section{Deferred Proofs of \Cref{sec:ex-post-alpha}}

\subsection{Proof of \Cref{lem:moment-decomp}} \label{app:moment-decomp}

\begin{proof}
We first prove the following tail inequality for all $t>0$:
\begin{equation}\label{eq:tail_estimate}
\Pr(R\ge3t)
\le
\Pr(\max_{i \in [n]} Y_i\ge t)+\Pr(R\ge t)^2.
\end{equation}
Let $S_i=\sum_{k\le i}Y_k$, with $S_0=0$. 
On the event $\{R\ge3t,\,\max_{i \in [n]} Y_i<t\}$, let $\tau$ be the first index such that
$S_\tau\ge t$.  Given that for $Y_i\leq t$ for all $i$, at the first crossing of level $t$, the overshoot is
less than $t$:
\begin{equation*}
S_\tau=S_{\tau-1}+Y_\tau< t+t=2t.
\end{equation*}
Since $R\ge3t$, the remaining mass after $\tau$ is larger than $t$:
\begin{equation*}
\sum_{j>\tau}Y_j=R-S_\tau>3t-2t=t.
\end{equation*}
For each fixed $i$, the event $\{\tau=i\}$ is determined by
$Y_1,\ldots,Y_i$, while the event $\left\{\sum_{j>i}Y_j>t\right\}$ is determined by later independent coordinates. Hence
\begin{equation*}
\Pr(R\ge3t,\,\max_{i \in [n]} Y_i<t)
\!\le\!
\sum_i
\Pr\big(\tau=i,\sum_{j>i}Y_j>t\big)  \!=\!
\sum_i
\Pr(\tau=i)
\Pr\big(\!\sum_{j>i}Y_j>t\big)\! \le\!
\Pr(R\ge t)
\sum_i\Pr(\tau=i),
\end{equation*}
where the first inequality comes from the fact that under the event studied, $\tau$ must stop due to $R\geq 3t \geq t$. In addition, the event $\{\tau<\infty\}$ is exactly $\{R\ge t\}$. Thus $\sum_i\Pr(\tau=i)=\Pr(\tau<\infty)=\Pr(R\ge t)$, and therefore $\Pr(R\ge3t,\,\max_{i \in [n]} Y_i<t)
\le
\Pr(R\ge t)^2$. Finally, 
\begin{equation*}
\Pr(R \geq 3t) = \Pr(R \geq 3t, \max_{i\in[n]} Y_i \geq t) +\Pr(R \geq 3t, \max_{i\in[n]} Y_i <t) \leq \Pr(\max_{i \in [n]} Y_i \geq t) +\Pr(R\geq t)^2.
\end{equation*}

We now prove the upper bound on the prophet value. We have
\begin{align*}
\be[R^q]
& =\int_0^\infty \Pr(R \geq u^{1/q}) du  \tag{Tail formula for expectation}\\
&=\int_0^\infty \Pr(R\ge3t)q3^q t^{q-1} dt \tag{Change of variable $u=3^q t^q$}\\
&\le
q3^q\int_0^\infty t^{q-1}\Pr(\max_{i \in [n]} Y_i\ge t)\,dt
+
q3^q\int_0^\infty t^{q-1}\Pr(R\geq t)^2\,dt \tag{\Cref{eq:tail_estimate}} \\
&=
3^q\be[\max_{i \in [n]} Y_i^q]
+
q3^q\int_0^\infty t^{q-1}\Pr(\min\{R,\widetilde R \} \geq t)\,dt \tag{Tail probability of two i.i.d. r.v.}\\
&=3^q (\be[\max_{i \in [n]} Y_i^q] + \be[\min\{ R^q, \widetilde R^q\}]). \qedhere
\end{align*}
\end{proof}

\subsection{Proof of \Cref{lem:exponential-race}} \label{app:exponential-race}

\begin{proof}
For simplicity, we assume that the $Y_i$ have finite support, the general case follows by a standard discretization argument.

Fix $t>0$ and $i$ with $\Pr^{(q)}(Y_i>0)>0$. By independence of the $E_i$ conditionally on $Y$,
\begin{align*}
g(y)\coloneqq \Pr^{(q)}\left(\min_{j\neq i}E_j\geq t\mid Y_i=y\right)
&=\be^{(q)}\left[
\Pr^{(q)}\left(E_j\geq t\text{ for every }j\neq i\mid Y\right)
\,|\,Y_i=y\right]\\
&=\be^{(q)}\big[
\prod_{j\neq i}\Pr^{(q)}(E_j\geq t\mid Y)
\,|\,Y_i=y\big]\\
&=\be^{(q)}\big[
\prod_{j\neq i}e^{-tY_j}
\,|\,Y_i=y\big]\\
&=\be^{(q)}[e^{-t(R-Y_i)}\mid Y_i=y].
\end{align*}
We show that $g$ is nondecreasing. The change of measure and independence under the original law give
\begin{align*}
\Pr^{(q)}(R-Y_i \in dr \mid Y_i=y)=\frac{\Pr^{(q)}(Y_i=y,R-Y_i \in dr )}
{\Pr^{(q)}(Y_i=y)}
&=\frac{(y+r)^q\Pr(Y_i=y,R-Y_i \in dr )}
{\sum_u (y+u)^q\Pr(Y_i=y,R-Y_i=u)}\\
&=\frac{(y+r)^q\Pr(R-Y_i \in dr \mid Y_i=y)}
{\sum_u (y+u)^q\Pr(R-Y_i=u\mid Y_i=y)}\\
&=\frac{(y+r)^q\Pr(R-Y_i \in dr \mid Y_i=y)}
{\be[R^q\mid Y_i=y]}\\
&=\frac{(y+r)^q\Pr(R-Y_i \in dr )}
{\be[(y+R-Y_i)^q]},
\end{align*}
where the last equality comes from the independence of $R-Y_i=\sum_{j \neq i} Y_j$ with $Y_i$.

Fix $y'>y>0$ and define $w(r)=\left(\frac{y'+r}{y+r}\right)^q$. 
The previous formula gives
\begin{align*}
\be^{(q)}[w(R-Y_i)\mid Y_i=y]
=\int_0^\infty \!\!
w(r)\frac{(y+r)^q}{\be[(y+R-Y_i)^q]}
\Pr(R-Y_i\in dr)
&=\frac{\int_0^\infty(y'+r)^q\Pr(R-Y_i\in dr)}
{\be[(y+R-Y_i)^q]}\\
&=\frac{\be[(y'+R-Y_i)^q]}{\be[(y+R-Y_i)^q]}.
\end{align*}
Therefore, the conditional law at $y'$ is obtained from that at $y$ by multiplying by $w$ and normalizing, $\Pr^{(q)}(R-Y_i\in dr\mid Y_i=y')
=\frac{w(r)}{\be^{(q)}[w(R-Y_i)\mid Y_i=y]}
\Pr^{(q)}(R-Y_i\in dr\mid Y_i=y)$.
Taking the conditional expectation of $e^{-t(R-Y_i)}$ against these conditional laws, we obtain
\begin{align*}
g(y')=\int_0^\infty e^{-tr}\Pr^{(q)}(R-Y_i\in dr\mid Y_i=y')&=\frac{\int_0^\infty e^{-tr}w(r)
\Pr^{(q)}(R-Y_i\in dr\mid Y_i=y)}
{\be^{(q)}[w(R-Y_i)\mid Y_i=y]}\\
&=\frac{\be^{(q)}[e^{-t(R-Y_i)}w(R-Y_i)\mid Y_i=y]}
{\be^{(q)}[w(R-Y_i)\mid Y_i=y]}.
\end{align*}
Hence the difference between $g(y)$ and $g(y')$ is equal to 
\begin{align*}
g(y')-g(y)&=\frac{
\be^{(q)}[e^{-t(R-Y_i)}w(R-Y_i)\mid Y_i=y]
-g(y)\be^{(q)}[w(R-Y_i)\mid Y_i=y]}
{\be^{(q)}[w(R-Y_i)\mid Y_i=y]}\\
&=\frac{\operatorname{Cov}^{(q)}
\left(e^{-t(R-Y_i)},w(R-Y_i)\mid Y_i=y\right)}
{\be^{(q)}[w(R-Y_i)\mid Y_i=y]}
\geq0.
\end{align*}
The covariance is positive because both $r \mapsto e^{-tr}$ and $w$ are decreasing in $r$ for $y'>y$. Hence $g$ is nondecreasing. If $\Pr^{(q)}(Y_i=0)>0$, the same conclusion at zero follows by dominated convergence. Because $g$ is nondecreasing, this further implies the inequality
\begin{equation*}
\be^{(q)}[Y_i g(Y_i)\mid E_i\geq t]
-\be^{(q)}[Y_i\mid E_i\geq t]\be^{(q)}[g(Y_i)\mid E_i\geq t]
=\operatorname{Cov}^{(q)}(Y_i,g(Y_i)\mid E_i\geq t)
\geq0.
\end{equation*}\medskip 

Given $Y$, the density of $E_i$ at $t$ is $Y_i e^{-tY_i}$ and its survival probability is $e^{-tY_i}$. First, we have
\begin{align*}
\be^{(q)}[Y_i\mid E_i\geq t]=\frac{\be^{(q)}[Y_i\mathbf{1}_{\{E_i\geq t\}}]}
{\Pr^{(q)}(E_i\geq t)}=\frac{\be^{(q)}\left[
\be^{(q)}[Y_i\mathbf{1}_{\{E_i\geq t\}}\mid Y_i]
\right]}
{\be^{(q)}[\Pr^{(q)}(E_i\geq t\mid Y_i)]}&=\frac{\be^{(q)}[Y_i\Pr^{(q)}(E_i\geq t\mid Y_i)]}
{\be^{(q)}[\Pr^{(q)}(E_i\geq t\mid Y_i)]}\\
&=\frac{\be^{(q)}[Y_i e^{-tY_i}]}
{\be^{(q)}[e^{-tY_i}]}.
\end{align*}
We now compare the law of $Y_i$ when conditioned with respect to $E_i=t$ and $E_i\geq t$. Using Bayes' rule and the previous equality,
\begin{align*}
\Pr^{(q)}(Y_i\in dy\mid E_i=t)
=\frac{ye^{-ty}\Pr^{(q)}(Y_i\in dy)}
{\be^{(q)}[Y_i e^{-tY_i}]}&=\frac{ye^{-ty}\Pr^{(q)}(Y_i\in dy)}
{\be^{(q)}[e^{-tY_i}]\be^{(q)}[Y_i\mid E_i\geq t]}\\
&=\frac{y}{\be^{(q)}[Y_i\mid E_i\geq t]}
\Pr^{(q)}(Y_i\in dy\mid E_i\geq t).
\end{align*}
Then, by independence  
\begin{align*}
\Pr^{(q)}(\min_{j\neq i}E_j\geq t\mid E_i=t)
=\be^{(q)}[
\Pr^{(q)}(\min_{j\neq i}E_j\geq t\mid Y_i,E_i=t)
\,|\,E_i=t]
&=\be^{(q)}[
\Pr^{(q)}(\min_{j\neq i}E_j\geq t\mid Y_i)
\,|\,E_i=t]\\
&=\be^{(q)}[g(Y_i)\mid E_i=t].
\end{align*}
We can finally compare the probability that the minimum of the other exponential variables is greater than $t$, depending on whether we know that $E_i =t$ or $E_i \geq t$. 
\begin{align*}
\Pr^{(q)}(\min_{j\neq i}E_j\geq t\mid E_i=t)
-\Pr^{(q)}(\min_{j\neq i}E_j\geq t\mid E_i\geq t)&=\be^{(q)}[g(Y_i)\mid E_i=t]
-\be^{(q)}[g(Y_i)\mid E_i\geq t]\\
&=\frac{\be^{(q)}[Y_i g(Y_i)\mid E_i\geq t]}
{\be^{(q)}[Y_i\mid E_i\geq t]}
-\be^{(q)}[g(Y_i)\mid E_i\geq t]\\
&=\frac{\operatorname{Cov}^{(q)}(Y_i,g(Y_i)\mid E_i\geq t)}
{\be^{(q)}[Y_i\mid E_i\geq t]}
\geq0.
\end{align*}
Finally, using this last inequality
\begin{equation*}
\Pr^{(q)}(J=i\mid E_i=t)
=\Pr^{(q)}\big(\min_{j\neq i}E_j\geq t\mid E_i=t\big)\geq\Pr^{(q)}\big(\min_{j\neq i}E_j\geq t\mid E_i\geq t\big)=\frac{F(t)}{\Pr^{(q)}(E_i\geq t)}
\geq F(t).
\end{equation*}
When $E_i=\infty$, the claim is immediate because $F(\infty)=0$.
\end{proof}

\subsection{Proof of \Cref{lem:alpha_prophet}} \label{app:alpha_prophet}

\begin{proof}
Let $x\in\br_+^n$. If some $x_i=0$, every feasible allocation gives value $0$, so assume $x_i>0$ for all $i\in[n]$. Since $1-\alpha<0$, maximizing the displayed objective is equivalent to minimizing $\sum_{i=1}^n(\alloc_i x_i)^{1-\alpha}$. For any feasible allocation $\alloc$, we have
\begin{align*}
\sum_{i=1}^n(\alloc_i x_i)^{1-\alpha}
&=
\sum_{i=1}^n \alloc_i^{1-\alpha}x_i^{1-\alpha} \\[1mm]
&=
\left(\sum_{j=1}^n x_j^{\frac{1-\alpha}{\alpha}}\right)^\alpha
\sum_{i=1}^n
\left(
\frac{x_i^{\frac{1-\alpha}{\alpha}}}
{\sum_{j=1}^n x_j^{\frac{1-\alpha}{\alpha}}}
\right)^\alpha
\alloc_i^{1-\alpha} \\
&=
\left(\sum_{j=1}^n x_j^{\frac{1-\alpha}{\alpha}}\right)^\alpha
\sum_{i=1}^n
\frac{x_i^{\frac{1-\alpha}{\alpha}}}
{\sum_{j=1}^n x_j^{\frac{1-\alpha}{\alpha}}}
\left(
\frac{\alloc_i}
{x_i^{\frac{1-\alpha}{\alpha}}/
\sum_{j=1}^n x_j^{\frac{1-\alpha}{\alpha}}}
\right)^{1-\alpha}.
\end{align*}
The weights $x_i^{\frac{1-\alpha}{\alpha}}
    /\sum_{j=1}^n x_j^{\frac{1-\alpha}{\alpha}}$ 
sum to $1$, and the map $t\mapsto t^{1-\alpha}$ is convex on $\br_+$ because $1-\alpha<0$. Hence Jensen's inequality gives
\begin{align*}
\sum_{i=1}^n
\frac{x_i^{\frac{1-\alpha}{\alpha}}}
{\sum_{j=1}^n x_j^{\frac{1-\alpha}{\alpha}}}
\left(
\frac{\alloc_i}
{x_i^{\frac{1-\alpha}{\alpha}}/
\sum_{j=1}^n x_j^{\frac{1-\alpha}{\alpha}}}
\right)^{1-\alpha}
&\geq
\left(
\sum_{i=1}^n
\frac{x_i^{\frac{1-\alpha}{\alpha}}}
{\sum_{j=1}^n x_j^{\frac{1-\alpha}{\alpha}}}
\frac{\alloc_i}
{x_i^{\frac{1-\alpha}{\alpha}}/
\sum_{j=1}^n x_j^{\frac{1-\alpha}{\alpha}}}
\right)^{1-\alpha} \\
&=
\left(\sum_{i=1}^n \alloc_i\right)^{1-\alpha} \geq 1,
\end{align*}
where the last inequality uses $\sum_i\alloc_i\leq1$ and $1-\alpha<0$. Therefore, $\sum_{i=1}^n(\alloc_i x_i)^{1-\alpha}
    \geq
    \left(\sum_{j=1}^n x_j^{\frac{1-\alpha}{\alpha}}\right)^\alpha$.
This proves that every feasible allocation has objective value at most
\begin{equation*}
\left(
    \frac1n
    \left(\sum_{j=1}^n x_j^{\frac{1-\alpha}{\alpha}}\right)^\alpha
    \right)^{\frac{1}{1-\alpha}}.    
\end{equation*}
This upper bound is achieved by the feasible allocation $\alloc_i^\star
    =
    x_i^{\frac{1-\alpha}{\alpha}}/
    \sum_{j=1}^n x_j^{\frac{1-\alpha}{\alpha}}$.
Indeed,
\begin{align*}
\sum_{i=1}^n(\alloc_i^\star x_i)^{1-\alpha}
&=
\sum_{i=1}^n
\left(
\frac{x_i^{\frac{1-\alpha}{\alpha}}x_i}
{\sum_{j=1}^n x_j^{\frac{1-\alpha}{\alpha}}}
\right)^{1-\alpha} \\[1mm]
&=
\sum_{i=1}^n
\left(
\frac{x_i^{1/\alpha}}
{\sum_{j=1}^n x_j^{\frac{1-\alpha}{\alpha}}}
\right)^{1-\alpha} \\[1mm]
&=
\left(\sum_{j=1}^n x_j^{\frac{1-\alpha}{\alpha}}\right)^{\alpha-1}
\sum_{i=1}^n x_i^{\frac{1-\alpha}{\alpha}} \\[1mm]
&=
\left(\sum_{j=1}^n x_j^{\frac{1-\alpha}{\alpha}}\right)^\alpha. \qedhere
\end{align*} 
\end{proof}

\newpage

\bibliographystyle{ACM-Reference-Format}
\bibliography{bibliography}

\end{document}